%% file: sj297.tex
\documentclass[reqno,tbtags]{amsart}
\usepackage{amssymb}
\usepackage{url}
\usepackage[square,numbers]{natbib}
\bibpunct[, ]{[}{]}{;}{n}{,}{,}
\usepackage{graphicx}

\title
{A unified approach to linear probing hashing with buckets}

\date{20 October, 2014}

\author{Svante Janson}
\thanks{SJ partly supported by the Knut and Alice Wallenberg Foundation}
\address{Department of Mathematics, Uppsala University, PO Box 480,
SE-751~06 Uppsala, Sweden}
\email{svante.janson@math.uu.se}
\newcommand\urladdrx[1]{{\urladdr{\def~{{\tiny$\sim$}}#1}}}
\urladdrx{http://www2.math.uu.se/~svante/}

\author{Alfredo Viola}
\address{Universidad de la Rep\'ublica, Montevideo, Uruguay}
\email{viola@fing.edu.uy}


\numberwithin{equation}{section}

\renewcommand\le{\leqslant}
\renewcommand\ge{\geqslant}

\allowdisplaybreaks

\newtheorem{theorem}{Theorem}[section]
\newtheorem{lemma}[theorem]{Lemma}

\newtheorem{corollary}[theorem]{Corollary}

 \theoremstyle{definition}

\newtheorem{remark}[theorem]{Remark}
\newtheorem*{acks}{Acknowledgements}

\newenvironment{alphenumerate}[1][0pt]{
\addtolength{\leftmargini}{#1}\begin{enumerate}
 }{\end{enumerate}}

\newenvironment{romenumerate}[1][0pt]{
\addtolength{\leftmargini}{#1}\begin{enumerate}
 }{\end{enumerate}}

\newcounter{oldenumi}
{\setcounter{oldenumi}{\value{enumi}}
\begin{romenumerate} \setcounter{enumi}{\value{oldenumi}}}
{\end{romenumerate}}

\newcounter{thmenumerate}

\newcommand\pfitemx[1]{\par#1:}
\newcommand\pfitemref[1]{\pfitemx{\ref{#1}}}

\newcommand{\refT}[1]{Theorem~\ref{#1}}
\newcommand{\refC}[1]{Corollary~\ref{#1}}
\newcommand{\refL}[1]{Lemma~\ref{#1}}
\newcommand{\refR}[1]{Remark~\ref{#1}}
\newcommand{\refS}[1]{Section~\ref{#1}}

\newcommand{\refF}[1]{Figure~\ref{#1}}

\newcommand{\refFig}[1]{Figure~\ref{#1}}

\newcommand{\sumj}{\sum_{j=0}^\infty}
\newcommand{\sumk}{\sum_{k=0}^\infty}
\newcommand{\sumki}{\sum_{k=1}^\infty}
\newcommand{\summ}{\sum_{m=0}^\infty}
\newcommand{\sumn}{\sum_{n=0}^\infty}
\newcommand{\sumni}{\sum_{n=1}^\infty}
\newcommand{\sumin}{\sum_{i=1}^n}

\newcommand{\prodin}{\prod_{i=1}^n}

\newcommand\set[1]{\ensuremath{\{#1\}}}

\newcommand\xpar[1]{(#1)}
\newcommand\bigpar[1]{\bigl(#1\bigr)}
\newcommand\Bigpar[1]{\Bigl(#1\Bigr)}
\newcommand\biggpar[1]{\biggl(#1\biggr)}
\newcommand\lrpar[1]{\left(#1\right)}

\newcommand\bigabs[1]{\bigl|#1\bigr|}

\def\rompar(#1){\textup(#1\textup)}    
\newcommand\xfrac[2]{#1/#2}

\newcommand\parfrac[2]{\lrpar{\frac{#1}{#2}}}

\newcommand\Bigparfrac[2]{\Bigpar{\frac{#1}{#2}}}

\def\xexp(#1){e^{#1}}

\newcommand\floor[1]{\lfloor#1\rfloor}

\newcommand\ntoo{\ensuremath{{n\to\infty}}}
\newcommand\Ntoo{\ensuremath{{N\to\infty}}}

\newcommand\bmin{\wedge}

\newcommand\punkt{.\spacefactor=1000}    
\newcommand\iid{i.i.d\punkt}    
\newcommand\ie{i.e\punkt}
\newcommand\eg{e.g\punkt}

\newcommand\cf{cf\punkt}
\newcommand{\as}{a.s\punkt}

\newcommand\ii{\mathrm{i}}

\newcommand{\tend}{\longrightarrow}
\newcommand\dto{\overset{\mathrm{d}}{\tend}}
\newcommand\pto{\overset{\mathrm{p}}{\tend}}

\newcommand\bbC{\mathbb C}
\newcommand\bbN{\mathbb N}

\newcommand\bbZ{\mathbb Z}

\newcounter{CC}
\newcommand{\CC}{\stepcounter{CC}\CCx} 
\newcommand{\CCx}{C_{\arabic{CC}}}     
\newcounter{cc}

\newcommand\E{\operatorname{\mathbb E{}}}
\newcommand\PP{\Pr}
\newcommand\Var{\operatorname{Var}}

\newcommand\Po{\operatorname{Po}}

\newcommand\Bin{\operatorname{Bin}}

\newcommand\ga{\alpha}
\newcommand\gb{\beta}
\newcommand\gd{\delta}

\newcommand\gL{\Lambda}
\newcommand\go{\omega}
\newcommand\gO{\Omega}
\newcommand\gs{\sigma}

\newcommand\eps{\varepsilon}

\renewcommand\phi{\xxx}  

\newcommand\etta{\boldsymbol1}

\newcommand\qw{^{-1}}
\newcommand\qww{^{-2}}
\newcommand\qq{^{1/2}}

\newcommand\dd{\,\mathrm{d}}
\newcommand\ddx{\mathrm{d}}

\newcommand{\pgf}{probability generating function}

\newcommand\rhs{right-hand side}

\newcommand\etto{\bigpar{1+o(1)}}

\newcommand\HT{\mathfrak{T}}
\newcommand\psih{\psi_H}
\newcommand\psiq{\psi_Q}
\newcommand\psix{\psi_X}
\newcommand\psib{\psi_B}
\newcommand\psiu{\psi_U}
\newcommand\psiv{\psi_V}
\newcommand\psicrh{\psi_C}
\newcommand\psitu{\psi_{\tU}}
\newcommand\psifc{\psi_{\FC}}
\newcommand\psirh{\psi_{\RH}}
\newcommand\psihb{\psi_{\hB}}
\newcommand\aea{\ga e^{-\ga}}
\newcommand\zetal{\zeta_\ell}
\newcommand\prodlb{\prod_{\ell=0}^{b-1}}
\newcommand\prodlbi{\prod_{\ell=1}^{b-1}}
\newcommand\ctt{\tau}
\newcommand\mni{_{m,n;i}}
\newcommand\gai{_{\ga;i}}
\newcommand\ioooo{_{i=-\infty}^\infty}
\newcommand\och[1]{;\, #1}
\newcommand\intoa{\int_0^\ga}
\newcommand\dfc{D^{\FC}}
\newcommand\dfcnn{\dfc_{n,n}}
\newcommand\drh{D^{\RH}}
\newcommand\crh{C^{\RH}}

\newcommand\FC{\mathsf{FC}}
\newcommand\RH{\mathsf{RH}}
\newcommand\tU{\tilde U}
\newcommand\zetaxy[1]{T\bigpar{\go^{#1}\ga e^{-\ga}}}
\newcommand\zetaxx[1]{\zetaxy{#1}/\ga}
\newcommand\zetalx{\zetaxx{\ell}}

\newcommand\zetaly{\zetaxy{\ell}}

\newcommand\tba{T_0(b\ga)}

\newcommand\Pos{\operatorname{Pos}}

\newcommand\Add{\operatorname{Add}}
\newcommand\Bucket{\operatorname{Bucket}}

\newcommand\Mark{\operatorname{Mark}}
\newcommand\Seq{\operatorname{Seq}}
\newcommand\LL{^{(L)}}
\newcommand\Lm{^{(m)}}
\newcommand\ppm{{\bf P}_m}
\newcommand\zq{q}
\newcommand\Combinatorial{Combinatorial approach}
\newcommand\Probabilistic{Probabilistic approach}
\newcommand\mn{_{m,n}}
\newcommand\taux{T_0(b\ga)}
\newcommand\cttx{\taux}
\newcommand\eqrefhq{\eqref{h}--\eqref{q}}
\newcommand\infpoi{infinite Poisson model}
\newcommand\asmn{as $m,n\to\infty$ with $n/bm\to\ga$}
\newcommand\hB{\hat B}
\newcommand\ddz{\frac{\partial}{\partial z}}
\newcommand\ddzx{\xfrac{\partial}{\partial z}}
\newcommand\ddw{\frac{\partial}{\partial w}}
\newcommand\Uq{\mathsf{U}_q}
\newcommand\HH{\mathsf{H}}
\newcommand\lhopitals{l'H\^opital's rule}

\newcommand{\Holder}{H\"older}

\begin{document}

\begin{abstract} 
We give a unified analysis of linear probing hashing with a general bucket
size. We use both a combinatorial approach, giving exact formulas for
generating functions, and a probabilistic approach, giving simple
derivations of asymptotic results. Both approaches complement nicely,
and give a good insight in the relation between linear probing and
random walks. A key methodological contribution, at the core of
Analytic Combinatorics, is the use of the 
symbolic method (based on $q$-calculus) to directly derive the generating
functions to analyze.
\end{abstract}

\maketitle

\section{Motivation}\label{S:motivation}

{\em Linear probing hashing\/}, defined below,
is certainly the simplest ``in place'' hashing
algorithm~\cite{KnuthIII}.
\begin{itemize}
\item[]
\sl A table of length $m$, $T[1\,.\,.\,m]$, with buckets of size $b$ is set up,
as well as a hash function $h$ that maps keys from some domain to
the interval $[1\,.\,.\,m]$ of table addresses.
A collection of $n$ keys with
$n\le bm$ are entered sequentially  into the table
according to the following rule: Each key $x$ is placed at
the first bucket that is not full starting from $h(x)$ in cyclic order,
namely
the first of $h(x),h(x)+1,\ldots,m,1,2,\ldots,h(x)-1$.
\end{itemize}

In \cite{Knuth98} Knuth motivates his paper in the following way: 
``The purpose of this note is to exhibit a surprisingly simple solution
to a problem that appears
in a recent book by Sedgewick and Flajolet \cite{SedFla96}:

\noindent
\textbf{Exercise 8.39}\; Use the symbolic method to derive the EGF of the number
of probes
required by linear probing in a successful search, for fixed $M$.''

Moreover, at the end of the paper in his personal remarks he declares:
``None of the methods available in 1962 were powerful enough to deduce the expected 
square displacement, much less the higher moments, so it is an even greater 
pleasure to be able to derive such results today from other work that has
enriched the field of combinatorial mathematics during a period of 35
years.'' In this sense, he is talking about the powerful methods based
on Analytic Combinatorics that has been developed for the last decades,
and are presented in \cite{FlaSed09}.

In this paper we present in a unified way the analysis of several
random variables related with linear probing hashing with buckets,
giving explicit and exact trivariate generating functions in the
combinatorial model, 
together with  generating functions in the asymptotic Poisson
model that provide limit results, and relations between the two types of
results. 
We consider also the \emph{parking problem} version, where there is no
wrapping around and overflow may occur from the last bucket.
Linear probing has been shown to have strong connections with
several important problems (see \cite{Knuth98,FlaPoVio,ChFl03} and the 
references therein). The derivations in the asymptotic
Poisson model are probabilistic and use heavily
the relation between random walks and the profile of the table.
Moreover, the derivations in the combinatorial model are based in
combinatorial specifications that directly translate into multivariate
generating functions. As far as we know, this is the first 
unified presentation of the analysis of linear probing hashing with buckets
based on
Analytic Combinatorics (``if you can specify it, you can analyze it'').

We will see that results can easily be translated between
the exact combinatorial model and the asymptotic Poisson model.
Nevertheless, we feel that it is important to present independently 
derivations for the two models, since the methodologies complement very
nicely. Moreover, 
they heavily rely in the deep relations between linear probing and
other combinatorial problems like random walks, and the power of
Analytic Combinatorics.

The derivations based on Analytic Combinatorics heavily rely on a
lecture presented by Flajolet whose notes can be accessed in
\cite{Flajolet:slides}. Since these ideas have only been partially 
published in the context of the analysis of hashing in 
\cite{FlaSed09}, we briefly present here some
constructions that lead to  $q$-analogs of their corresponding
multivariate generating functions.

\section{Some previous work}\label{S:previous}
The main application of linear probing is to retrieve information in
secondary storage devices when the load factor is not too high,
as first proposed by Peterson \cite{Peterson}.
One reason for the use of linear probing is that it
preserves locality of reference between successive probes, thus avoiding
long seeks \cite{Larson}.

The first published analysis of linear probing was
done by Konheim and Weiss \cite{KW}.
In addition, this problem also has a special historical value 
since the first analysis of algorithms ever performed by D. Knuth \cite{Knuth63}
was that of linear probing hashing.
As Knuth indicates in many of his writings, the problem has had a
strong influence on his scientific carreer.
Moreover, the construction cost to fill a linear probing hash table
connects to a wealth of interesting combinatorial and analytic problems.
More specifically, the Airy distribution that surfaces as a limit law
in this construction cost is also present in random trees
(inversions and path length), random graphs (the complexity or
excess parameter), and in random walks (area), 
as well as in Brownian motion (Brownian excursion area)
\cite{Knuth98,FlaPoVio,SJ201}.

Operating primarily in the context of double hashing,
several authors \cite{Brent,Amble,GM} observed that a collision could be
resolved in favor of {\it{any}} of the keys involved, and used this
additional degree of freedom to decrease the expected search time in the
table. We obtain the standard scheme by letting the incoming key probe its
next location. So, we may see this standard policy as a
{\it first-come-first-served} (FCFS) heuristic.
Later Celis, Larson and Munro \cite{CelisT,Celis} were the first to observe that
collisions could be resolved having {\it{variance reduction}} as a goal.
They defined the Robin Hood heuristic, in which each collision occurring on
each insertion is resolved in
favor of the key that is farthest away from its home location. Later,
Poblete and Munro \cite{LCFS} defined the {\it last-come-first-served}
(LCFS) heuristic,
where collisions are resolved in favor of the incoming key, and others are
moved ahead one position in their probe sequences. These strategies do
not look ahead in the probe sequence, since the decision is made before
any of the keys probes its next location. As a consequence, they do not
improve the average search cost, although the variance of this random
variable is different for each strategy.

For the FCFS heuristic, if $A_{m,n}$ denotes the number
of probes in a successful search in a hash table of size $m$
with $n$ keys
(assuming all keys in the table
are equally likely to be searched), and if we assume
that the hash function $h$ takes all the
values in $0\ldots m-1$ with equal probabilities, then we know
from
\cite{KnuthIII,GoBa91}
\begin{align}
 {\E}[A_{m,n}] &= \frac{1}{2}(1+Q_0(m,n-1)),
\label{average} 
\\
 {\Var}[A_{m,n}] 
&=
\frac{1}{3}Q_2(m,n-1)-\frac{1}{4}Q_0(m,n-1)^2-\frac{1}{12}.
\end{align}
where
\begin{align*}
Q_r(m,n) = \sum_{k=0}^n \binom{k+r}{k}
\frac{n^{\underline{k}}}{m^k}
\end{align*}
and $n^{\underline{k}}$ defined as
$n^{\underline{k}} = n(n-1) \ldots (n-k+1)$
for real $n$ and integer $k \geq 0$ is the $k$:{th} falling
factorial power of $n$.
The function $Q_0(m,n)$ is also known as  Ramanujan's
$Q$-function \cite{RamanujanQ}.

For a table with $n = \alpha m$ keys, and fixed $\alpha < 1$
and $n,m \rightarrow \infty$, these quantities depend (essentially) only
on $\alpha$:
\begin{align*}
{\E}[A_{m,\alpha m}] &=
\frac{1}{2}\left(1+\frac{1}{1-\alpha}\right) -
\frac{1}{2(1-\alpha)^3m} + O\left(\frac{1}{m^2}\right),
\\
{\Var}[A_{m,\alpha m}] &= \frac{1}{3(1-\alpha)^3} -
\frac{1}{4(1-\alpha)^2} - \frac{1}{12} -
\frac{1+3\alpha}{2(1 - \alpha)^5m} +
O\left(\frac{1}{m^2}\right).
\end{align*}
For a full table, these approximations are useless, but the
properties of the
$Q$ functions can be used to obtain the following expressions, 
reproved in \refC{CFC0},
\begin{align}
{\E}[A_{m,m}]& = \frac{\sqrt{2\pi m}}{4}+\frac{1}{3}+
\frac{1}{48}\sqrt{\frac{2\pi}{m}}+O\left(\frac{1}{m}\right),
\label{ammE} \\
{\Var}[A_{m,m}] &= \frac{\sqrt{2\pi m^3}}{12}+
\left(\frac{1}{9}-\frac{\pi}{8}\right)m+ \frac{13\sqrt{2\pi
m}}{144}-
\frac{47}{405} - \frac{\pi}{48} +
O\left(\frac{1}{\sqrt{m}}\right).
\label{ammV}
\end{align}

As it can be seen the variance is very high, and as a
consequence the Robin Hood and LCFS heuristics are important in
this regard. It is proven in \cite{CelisT,Celis} that Robin Hood
achieves the minimum variance among all the heuristics that do
not look ahead at the future, and that LCFS has an
asymptotically optimal variance \cite{DPT}. This problem 
appears in the simulations presented in Section \ref{SER}, and
as a consequence even though the expected values that we present
are very informative, there is still a disagreement with the
experimental results.

Moreover, in \cite{SJ157} and \cite{Exact}, a distributional
analysis for the FCFS, LCFS and Robin Hood heuristic is
presented.
These results consider a hash table with buckets of size 1.
However, very little is known when we have tables with buckets
of size $b$.

In \cite{Bigbuck}, Blake and Konheim studied the asymptotic
behavior
of the expected cost of successful searches
as the number of keys and buckets tend to infinity
with their ratio remaining constant. Mendelson \cite{Mendel}
derived
exact formulae for the same expected cost, but only solved them
numerically.
These papers consider the FCFS heuristic. Moreover, in \cite{RHBuckets}
an exact analysis of a linear probing hashing scheme with buckets
of size $b$ (working with the Robin Hood heuristic) is presented.
The first complete distributional analysis of the Robin Hood
heuristic with buckets of size $b$ is presented in \cite{Viola},
where an independent analysis of the parking problem presented
in \cite{SeitzDiploma} is also proposed.

In the present paper we consider an arbitrary bucket size $b\ge1$. 
The special case $b=1$ has been studied in many previous works.
In this case, many of our results reduce to known results, see \eg{} 
\cite{SJ157} and \cite{Exact} and the references given there.

\section{Some notation}
We study tables with $m$ buckets of size $b$ and $n$ keys, where $b\ge1$
is a  constant. We often consider limits as $m,n\to\infty$ with $n/bm\to\ga$
with $\ga\in(0,1)$. We consider also the Poisson model with
$n\sim\Po(\ga bm)$, and thus $\Po(b\ga)$ keys hashed to each bucket; in
this model we can also take $m=\infty$ which gives a natural limit object, see
Sections \ref{Shash}--\ref{Sconv}.


A \emph{cluster} or \emph{block} is a (maximal) sequence of full
buckets ended by a non-full one.
An \emph{almost full table} is a table consisting of a single cluster.

The \emph{tree function}
\cite[p.~127]{FlaSed09}
is defined by
\begin{equation}
  \label{tree}
T(z) := \sum_{n=1}^\infty \frac{n^{n-1}}{n!} z^n,
\end{equation}
which converges for $|z|\le e^{-1}$;
$T(z)$ 
satisfies $T(e\qw)=1$ and
\begin{equation}\label{ztt}
  z = T(z)e^{-T(z)}.
\end{equation}
In particular, note that \eqref{ztt} implies that $T(z)$ is injective 
on $|z|\le e\qw$.
Recall also the well-known formula (easily obtained by
taking the logarithmic derivative of \eqref{ztt})
\begin{equation}\label{T'}
  T'(z)=\frac{T(z)}{z(1-T(z))}.
\end{equation}
(The tree function is related to the
Lambert $W$-function $W(z)$ \cite[\S 4.13]{NIST} by
$T(z)=-W(-z)$.) 

Let 
\begin{math}
  \go=\go_b:=e^{2\pi\ii/b}
\end{math}
be a primitive $b$:th unit root.

For $\ga$ with $0<\ga<1$, we define 
\begin{equation}\label{zetalzq}
  \zetal(\zq) 
=  \zetal(\zq;\ga) 
:=  T\bigpar{\go^\ell\aea \zq^{1/b}}/\ga
\end{equation}
and
\begin{equation}\label{zetal}
  \zetal :=\zetal(1) =\zetalx.
\end{equation}
Note that
\begin{equation}\label{zeta0}
  \zeta_0 =
T\bigpar{\ga e^{-\ga}}/\ga = 1,
\end{equation}
\cf{} \eqref{ztt} (with $z=\ga e^{-\ga}$).
We note the following properties of these numbers.

\begin{lemma}
  \label{Lzetalzq} 
Let\/ $0<\ga<1$. 
Then \eqref{zetalzq} defines
$b$ numbers $\zetal(\zq)$, $\ell=0,\dots,b-1$,
for every $\zq$ with $|\zq|\le R:=(e^{\ga-1}/\ga)^b>1$.
If furthermore $\zq\neq0$, then these $b$ numbers are distinct.

If\/ $|\zq|\le1$, then the $b$ numbers $\zetal(\zq)$, $\ell=0,\dots,b-1$,
satisfy $|\zetal(\zq)|\le1$, and they 
are the $b$ roots in the unit disc of 
\begin{equation}\label{zetalzq2}
  \zeta^b=e^{\ga b(\zeta-1)}\zq.
\end{equation}
\end{lemma}

\begin{proof}
Note first that 
\begin{equation}
  |\aea\go^\ell \zq^{1/b}|\le e\qw
\end{equation}
is equivalent to
\begin{equation}
  |\zq|\le R:=(e^{\ga-1}/\ga)^b>1,
\end{equation}
so all $\zetal(\zq)$ are defined for $|\zq|\le R$.
If also $\zq\neq0$, then 
$\zeta_0(\zq),\dots,\zeta_{b-1}(\zq)$ are distinct,
because $T$ is injective by \eqref{ztt}.
Furthermore, for $|\zq|\le1$,
using \eqref{zetalzq}, the fact that \eqref{tree} has positive coefficients,
and \eqref{zeta0},
\begin{equation}\label{hugo}
  |\zetal(\zq)|\le \ga\qw T\bigpar{\aea}=1.
\end{equation}
Moreover,  
\eqref{zetalzq} and \eqref{ztt} imply
\begin{equation}
  \ga\zetal(\zq) e^{-\ga\zetal(\zq)}=\aea\go^\ell \zq^{1/b}
\end{equation}
which by taking the $b$:th powers yields \eqref{zetalzq2}.
Since the derivative of $e^{\ga b (r-1)}-r^b$ at $r=1$ is $\ga b-b<0$,
we can find $r>1$ such that $e^{\ga b(r-1)}<r^b$, and then 
Rouch\'e's theorem shows that, for any $\zq$ with $|\zq|\le1$, 
$\zeta^b-e^{\ga b(\zeta-1)}\zq$ has exactly $b$ roots in $|\zeta|<r$; 
thus, the $b$ roots
$\zeta_0(q),\dots,\zeta_{b-1}(q)$ are the only roots of \eqref{zetalzq2}
in $|\zeta|<r$. (The case $\zq=0$ is trivial.)  
\end{proof}

\begin{remark}
In order to define an individual $\zetal(\zq)$, we have to fix a
choice of $\zq^{1/b}$ in \eqref{zetalzq}. It is thus impossible to define each
$\zetal(\zq)$ as a continuous function in the entire unit disc; they rather are
different branches of a single, multivalued, function, with a branch point
at $0$. Nevertheless, it is
only the collection of all of them together that matters, for example in
\eqref{psib}, and this collection is uniquely defined.  
\end{remark}

We denote convergence in distribution of random variables by  $\dto$.

\section{Combinatorial characterization of linear probing}\label{S:combin}

As a combinatorial object, a non-full linear probing hash table is a sequence of
almost full tables (or clusters) \cite{Knuth98,FlaPoVio,Viola}. 
As a consequence, any random variable related with the
table itself (like block lengths, or the overflow in the parking problem)
or with a random key (like its search cost) can be studied in a
cluster (that we may assume to be the last one in the sequence), and then 
use the sequence construction.  \refF{decomp} presents an 
example of such a decomposition.

We briefly recall here some of the definitions presented in 
\cite{Bigbuck,Viola}.
Let $F_{bi+d}$ be the number of ways to construct an almost full table of
length
$i+1$ and size $bi+d$ (that is, there are $b-d$ empty slots in the last
bucket). Define also
\begin{align}
F_d(u) := \sum_{i\geq 0} F_{bi+d} \frac{u^{bi+d}}{(bi+d)!},&&&
N_d(z,w): =
\sum_{s=0}^{b-1-d} w^{b-s} F_s(zw),\quad 0\leq d \leq b-1. 
\label{laFd} 
\end{align}
In this setting $N_d(z,w)$ is the generating function for the number of
almost full tables with more than $d$ empty locations in the last bucket.
More specifically $N_0(z,w)$ is the generating function for the
number of all the almost full tables. 
(Our generating functions use the weight $w^{bm}z^n/n!$ for a table of
length $m$ and $n$ keys; they are thus exponential generating functions in
$n$
and ordinary generating functions in $m$.) 
We present below some basic identities.
\begin{lemma}
\begin{align} 
F(bz,x)& := \sum_{d=0}^{b-1} F_d(bz) x^d 
= x^b - \prod_{j=0}^{b-1} \left(x -\frac{T(\omega^jz)}{z}\right) 
,\label{Fd}
\\
\sum_{d=0}^{b-1} N_d(bz,w) x^d 
&= 
\frac{\displaystyle{\prod_{j=0}^{b-1} \left(1 -x\frac{T(\omega^jzw)}{z}\right)
-\prod_{j=0}^{b-1} \left(1 -\frac{T(\omega^jzw)}{z}\right)}}
{1-x}, \label{Nd} \\
\sum_{d=0}^{b-1} N_d(b\alpha,e^{-\alpha}) x^d 
&= 
\prod_{j=1}^{b-1} \left(1 -x\frac{T(\omega^j\alpha e^{-\alpha} )}{\alpha }\right).
\label{Nda}
\end{align}
\end{lemma}
\begin{proof}
Equation (\ref{Fd}) can be derived from Lemma 2.3 in \cite{Bigbuck}.

Moreover, from equation (\ref{laFd}), 
\begin{align*}
\sum_{d=0}^{b-1} N_d(bz,w) x^d &=
\sum_{d=0}^{b-1} x^d
\sum_{s=0}^{b-1-d} w^{b-s} F_s(bzw) =
\sum_{s=0}^{b-1} w^{b-s} F_s(bzw)
\sum_{d=0}^{b-1-s} x^d \\ 
&=
\frac{\displaystyle{\sum_{s=0}^{b-1} w^{b-s} F_s(bzw)
-\sum_{s=0}^{b-1} (wx)^{b-s} F_s(bzw)}}{1-x} \\
&=
\frac{\displaystyle{w^{b} F\left(bzw,\frac{1}{w}\right)
-(wx)^{b} F\left(bzw,\frac{1}{wx}\right) }}{1-x}.
\end{align*}

Then (\ref{Nd}) follows from (\ref{Fd}).

Finally \eqref{Nda} follows from \eqref{Nd}, and since 
$T(\alpha e^{-\alpha})=\alpha$
the factor for $j=0$ cancels the second product, and
gives the factor $(1-x)$ in the first one that cancels with the
denominator.
\end{proof}
\begin{corollary}
\begin{equation} 
N_d(bz,w) 
= [x^d]
\frac{\displaystyle{\prod_{j=0}^{b-1} \left(1 -x\frac{T(\omega^jzw)}{z}\right)
-\prod_{j=0}^{b-1} \left(1 -\frac{T(\omega^jzw)}{z}\right)}}
{1-x}, 
\label{Nd1}
\end{equation}
and more specifically (formula (3.8) in \cite{Bigbuck}),
\begin{equation} 
N_0(bz,w)
= 1 - \prod_{j=0}^{b-1} \left(1 -\frac{T(\omega^jzw)}{z}\right). 
\label{N0}
\end{equation}

Moreover (Lemma 5 in \cite{Viola}),
\begin{equation} 
N_d(b\alpha,e^{-\alpha})
= 
\left[ x^d \right] \prod_{j=1}^{b-1} \left(1 -x\frac{T(\omega^j\alpha e^{-\alpha}
)}{\alpha }\right).
\end{equation}
\end{corollary}

Let also $Q_{m,n,d}$, for $0\le d\le b-1$,
be the number of ways of inserting $n$ keys
into a table with $m$ buckets of size $b$, so that a given (say the
last) bucket of the table contains more than $d$ empty slots.
(We include the empty table, and define $Q_{0,0,d}=1$.)
In this setting, by a direct application of the sequence construction as
presented in \cite{FlaSed09} (sequence of almost full tables)
we derive a result presented in \cite{Bigbuck}:
\begin{equation}
\Lambda_0(bz,w) := \sum_{m \geq 0} \sum_{n \geq 0} Q_{m,n,0}
\frac{(bz)^n}{n!} w^{bm}
= \frac{1}{1-N_0(bz,w)}
=  \frac{1}{\prod_{j=0}^{b-1} \left(1
-\frac{T(\omega^jzw)}{z}\right)},
\label{laLambda0}
\end{equation}
where $\Lambda_0(bz,w)$ is the generating function for the number of ways
to construct hash tables such that their last bucket is not full,
and more generally
\begin{equation}
\Lambda_d(bz,w) := \sum_{m \geq 0} \sum_{n \geq 0} Q_{m,n,d}
\frac{(bz)^n}{n!} w^{bm}
= 1 + \frac{N_d(bz,w)}{1-N_0(bz,w)}.
\label{laLambdad}
\end{equation}


Moreover
$O_d(bz,w)$, the generating function for the number of ways
to construct hash tables such that their last bucket has exactly $d\le b-1$
keys, is
\begin{equation}
\label{LaOgen}
O_d(bz,w) := \frac{F_d(bzw)w^{b-d}}{1-N_0(bz,w)}.
\end{equation}

\begin{figure}[htbp]
  \begin{center}
    \noindent\hrule
    \medskip
   \includegraphics[clip=true,bb=100 510 500 665,width=0.65\textwidth]{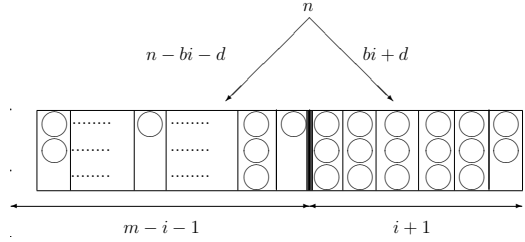}
    \caption{\small A decomposition for $b = 3$ and  $d = 2$.}
    \label{decomp}
    \medskip
    \noindent\hrule
  \end{center}
\end{figure}

Consider a hash table of length $m$ and $n$ keys, where
collisions are resolved by linear probing.
Let $P$ be a non-negative integer-valued
property (e.g.\ cost of a successful search or block
length), related with the last cluster of the sequence, or with a 
random key inside it.
Let $p_{bi+d}(q)$ be the probability generating function of $P$
calculated in the cluster of length $i+1$ and with $bi+d$ keys 
(as presented in \refF{decomp}).
We may express $p_{m,n}(q)$, the
generating function of $P$ for a table of length $m$ and $n$ keys 
with at least one empty spot in the last bucket,
as the sum of the conditional probabilities:
\begin{equation*}
p_{m,n}(q) = \sum_{d=0}^{b-1} \sum_{i \ge 0}
\#\set{\text{tables where last cluster has size $i+1$ and $bi+d$ keys}} 
~ p_{bi+d}(q).   
\end{equation*}

There are $Q_{m-i-1,n-bi-d,0}$ ways to insert $n-bi-d$ keys
in the leftmost hash table of length $m-i-1$, leaving their rightmost 
bucket not full. Moreover, there are $F_{bi+d}$ ways
to insert $bi+d$ keys in the almost full table of length $i+1$.
Furthermore, there are
$\binom{n}{bi+d}$ ways to choose which $bi+d$ keys go to the
last cluster.
Therefore,
\begin{equation*}
p_{m,n}(q) = \sum_{d=0}^{b-1} \sum_{i\geq0} 
\binom{n}{bi+d}
{Q_{m-i-1,n-bi-d,0}~F_{bi+d}}~~p_{bi+d}(q).
\end{equation*}

Then, the trivariate generating function for $p_{m,n}(q)$ is
\begin{align}
\label{Gen-mn1}
P(z,w,q) &:= \sum_{m,n\geq0}p_{m,n}(q)~w^{bm} \frac{z^n}{n!} 
= \Lambda_0(z,w) \hat{N}_0(z,w,q), 
\intertext{with}
\hat{N}_0(z,w,q) &:=
\sum_{i\geq 0} w^{b(i+1)}
\sum_{d=0}^{b-1}
F_{bi+d}~\frac{z^{bi+d}}{(bi+d)!}~p_{bi+d}(q),
\label{Gen-mn2}
\end{align}
which could be directly derived with the sequence construction
\cite{FlaSed09}. In this setting, $w$ marks the total capacity
($b(i+1)$) while $z$ marks the number of keys in the table
($bi+d$ with $0 \leq d < b$).
Equation \eqref{Gen-mn1} can be interpreted as follows: to
analyze a property in a linear probing hash table, do the
analysis in an almost full table
(giving the factor $\hat{N}_0(z,w,q)$), and then use the sequence construction
(giving the factor $\Lambda_0(z,w)$).

Notice that, as expected, $\hat{N}_0(z,w,1) = N_0(z,w)$ 
and $P(z,w,1) = \Lambda_0(z,w) -1$, 
since we consider only $m \geq 1$ (we have a last, non-filled bucket).

\begin{remark}\label{RGen}
  We have here, for simplicity, assumed that each $p_{bi+d}(q)$ is a
  probability generating function, corresponding to a single property $P$ of
  the last cluster. However, the argument above generalizes to the case when
$p_{bi+d}(q)$ is a generating function corresponding to several values of a
  property $P$ for each cluster; \eqref{Gen-mn1} and \eqref{Gen-mn2} still
  hold, although we no longer have $p_{bi+d}(1)=1$ and thus 
not $\hat{N}_0(z,w,1) = N_0(z,w)$ (in general).
We use this in Sections \ref{SU} and \ref{SFC}.
\end{remark}

\subsection{The Poisson Transform} \label{SPoissonTransform}
There are two standard models that are extensively used in the
analysis of
hashing algorithms: the {\it exact filling} model and the
{\it Poisson filling} model.
Under the exact model, we have a fixed number of keys, $n$, that
are
distributed among $m$ buckets of size $b$, and all $m^n$ possible
arrangements are
equally likely to occur.

Under the Poisson model, we assume that each location receives a
number of keys
that is Poisson distributed with parameter $b \alpha$ (with
$0\leq \alpha < 1$) , and is
{\em independent} of
the number of keys going elsewhere. This implies that the total
number
of keys, $N$, is itself a Poisson distributed random variable
with
parameter $b \alpha m$:
\begin{align*}
\Pr\left[N=n\right] = \frac{e^{-b \alpha m}(b \alpha m)^n}{n!},
\qquad n\ = \ 0,1,\ldots
\end{align*}
(For finite $m$, there is an exponentially small probability that $N>bm$; 
this possibility
can usually be ignored in the limit $m,n\to\infty$. 
For completeness we assume either that we consider the parking problem where
overflow may occur, or that we make some special definition when $N>bm$.
See also the infinite
Poisson model in \refS{Sconv}, where this problem disappears.)
This model was first considered in the analysis of hashing, 
for a somewhat different problem,
by Fagin {\it{et al}} \cite{Fagin} in 1979. 

The results obtained
under the Poisson filling model can be interpreted as an
approximation
of those one would obtain under the exact filling model when
$n=b\alpha m$.
For most situations and applications,
this approximation is satisfactory for large $m$ and $n$. 
(However, it cannot be used when we have
a full, or almost full, table, so $\alpha$ is very close to 1.)
We give a detailed statement of one general limit theorem (\refT{Tconv})
later, but give a brief sketch here. We begin with some algebraic identities.

Consider a hash table of size $m$ with $n$ keys, in which
conflicts are
resolved by open addressing using some heuristic.
Let the variable $P$ be a non-negative integer-valued property
of the table (e.g. the block length of a random cluster), or
of a random key of the table (e.g., the cost of a successful search),
and let $f_{m,n}$ be the result of applying a linear
operator $f$
(e.g., an expected value) to the probability generating function
of $P$ for the exact filling model.
Let
${\bf{P}}_m[f_{m,n};b\alpha]$ be
the result of applying the same linear operator $f$ to the
probability generating function of $P$ computed using 
the Poisson filling model. 
Then
\begin{align}
{\bf{P}}_m[f_{m,n};b\alpha]
&= \sum_{n \geq 0} \Pr\left[N = n \right] f_{m,n}
= e^{-bm\alpha}\sum_{n\geq 0} \frac{(bm\alpha)^n}{n!} f_{m,n}. \label{pt}
\end{align}
In this context, 
${\bf{P}}_m[f_{m,n};b\alpha]$ is called
the {\it{Poisson transform}} 
of $f_{m,n}$.  

In particular,
let $P_{m,n}(q)$ be the generating function
of a variable $P$ in a hash table of size $m$ with $n$ keys, and
let $p_{m,n}(q)=P_{m,n}(q)/m^n$ be the corresponding probability generating
  function of $P$ regarded as a random variable (with all $m^n$ tables
  equally likely). 
Define the trivariate generating function
\begin{equation*}
P(z,w,q) := 
\sum_{m\geq 0} w^{bm} \sum_{n\geq 0} {P_{m,n}(q)} \frac{z^n}{n!}
= 
\sum_{m\geq 0} w^{bm} \sum_{n\geq 0} p_{m,n}(q) \frac{(mz)^n}{n!}.
\end{equation*}
Then, for a fixed $0 \leq \alpha < 1$, using \eqref{pt},
\begin{equation}\label{laPoiss}
  \begin{split}
P(b\alpha,y^{1/b}e^{-\alpha},q) 
&= \sum_{m\geq 0} y^{m} \left(
e^{-bm\alpha} \sum_{n\geq 0} {p_{m,n}(q)}
\frac{(bm\alpha)^n}{n!}\right)
\\
&= \sum_{m\geq 0} y^{m} 
{\bf P}_m\left[{p_{m,n}(q)};b\alpha\right]. 	
  \end{split}
\end{equation}
In other words, 
$P(b\alpha,y^{1/b}e^{-\alpha},q) $ is the generating function of
$\ppm[p_{m,n}(q);b\ga]$.

Asymptotic results for the probability generating function in the Poisson
model,
can thus be found
by singularity analysis 
\cite{FlaOdl90,FlaSed09}
from $P(b\alpha,y^{1/b}e^{-\alpha},q)$. In
our problems, the dominant singularity is a simple pole at $y=1$.
Furthermore, asymptotic results for the exact model can be found by
de-Poissonization; we give a probabilistic proof of one such result 
in \refT{Tconv}\ref{tconvmn}.

The same formulas without the variable $q$ hold
if $P_{m,n}$ is the number of hash tables with a certain Boolean property, and
$p_{m,n}=P_{m,n}/m^n$ is the corresponding probability that a random hash table
has this property. 
In this case, the trivariate generating function \eqref{laPoiss} is replaced
by the bivariate
\begin{equation}\label{laPoiss2}
  \begin{split}
P(b\alpha,y^{1/b}e^{-\alpha}) 
&= \sum_{m\geq 0} y^{m} 
e^{-bm\alpha} \sum_{n\geq 0} {p_{m,n}}
\frac{(bm\alpha)^n}{n!}
.
  \end{split}
\end{equation}

For example, from equation (\ref{laLambda0}),
$\Lambda_0(b\alpha,y^{1/b}e^{-\alpha})$ has a dominant simple pole at $y=1$ 
originated by the factor with $j=0$,
since $T(\alpha e^{-\alpha})/\alpha = 1$. More precisely,
using \eqref{T'},
\begin{equation*}
\frac{1}{1-\frac{T\left(y^{\xfrac{1}{b}} \alpha
e^{-\alpha}\right)}{\alpha}} \quad^{~\displaystyle{\sim}}_{y \to 1}
~\frac{b(1-\alpha)}{1-y}.
\end{equation*} 
 and  the residue of 
$\Lambda_0(b\alpha,y^{1/b}e^{-\alpha})$ at $y=1$ is
\begin{equation}\label{T00}
 T_0(b\alpha)
:= \frac{b(1-\alpha)}{\prod_{j=1}^{b-1}
\left(1 - \frac{T(\go^j\alpha e^{-\alpha})}{\alpha} \right)}.  
\end{equation}
Then the following result presented in
(\cite{Bigbuck,Viola}) is rederived, 
see also \refT{Tconv}\ref{tconvPm}\ref{tconvres}: 
\begin{equation}\label{T0}
\lim_{m\to\infty}   
{\mathbf{P}}_m[Q_{m,n,0}/m^n;b\alpha] 
= T_0(b\alpha)
.  
\end{equation}
Moreover, from \eqref{Gen-mn1} we similarly obtain, for a property $P$,
\begin{align}
\lim_{m\to \infty} \mathbf{P}_m[p_{m,n}(q)/m^n;b\alpha]  &= T_0(b\alpha)
\hat{N}_0(b\alpha,e^{-\alpha},q).  \label{Ptrans}
\end{align}
By de-Poissonization, we further find asymptotics for $Q_{m,n,0}$ and 
(for suitable properties $P$) $p_{m,n}(q)$, 
see \refT{Tconv}\ref{tconvmn}.

Even though by equations \eqref{laPoiss} and \eqref{Ptrans}
the results in the exact model can be directly
translated into their counterpart in  the Poisson model, in this paper we
present derivations for both approaches. We feel this is very important to
present a unified point of view of the problem. Furthermore, the
deriviations made in each model are also unified. For the exact
model, a direct approach using the symbolic method and the
powerful tools from Analytic Combinatorics \cite{FlaSed09} is
presented, while for the Poisson model, a unified approach using
random walks is used. Presenting both related but independently
derived analyses, helps in the better understanding of the
combinatorial, analytic and probabilistic properties of linear
probing.

\section{A $q$-calculus to specify random variables}
\label{q-calculus}

All the generating functions in this paper are exponential in $n$ and
ordinary in $m$. Moreover, the variable $q$ marks the value of
the random variable at hand.
As a consequence all the labelled constructions in \cite{FlaSed09}
and their respective translation into EGF can be used. However, to
specify the combinatorial properties related with the analysis of
linear probing hashing, new constructions have to be added.
These ideas have been presented by Flajolet in \cite{Flajolet:slides},
but they do not seem to have been published in the context of hashing.
As a consequence, we briefly summarize them in this section.

\begin{figure}
\begin{center}
\begin{tabular}{|l|l|}
\hline
\underline{Marking a position} $\mapsto$ $\partial_w$ &
$C_{bn+d}=(n+1)A_{bn+d}$\\
$\mathcal{C} = \Pos(\mathcal{A}$) & $C(z,w)=\frac{w}{b}\frac{\partial}{\partial
w}(A(z,w))$  \\
\hline
\underline{Adding a key} $\mapsto$ $\int$ &
$C_{bn+d}=A_{{bn+d}-1}$  \\
$\mathcal{C} = \Add(\mathcal{A}$)  & $C(z,w)=\int_0^z A(u,w) du$  \\
\hline
\underline{Bucketing} $\mapsto$ $\exp$ & $C_{m,n} = 
\delta(m,1)$ \\
 $\mathcal{C} = \Bucket(\mathcal{Z}$)  & $C(z,w) = w^b\exp(z)$  \\
\hline
\underline{Marking a key}  $\mapsto$ $\partial_z$ & $C_{m,n} = 
n A_{m,n}$ \\
$\mathcal{C} = \Mark(\mathcal{A}$)  &
$C(z,w)=z\frac{\partial}{\partial z}A(z,w)$  \\
\hline
\end{tabular}
\caption{Constructions used in hashing}
\label{newcons}
\end{center}
\end{figure}

We first concentrate in counting generating functions ($q=1$),
and we generalize these constructions in Section \ref{qanalogues} for distributional
results where their $q$-analogue counterparts are presented.
\refFig{newcons} presents a list of combinatorial constructions used
in hashing and their corresponding translation into EGF, where
$\mathcal{Z}$ is an atomic class comprising a single key of size 1.
We motivate these constructions that are specifically defined for
this analysis.

Because of the sequence interpretation of linear probing,
insertions are done in an almost full table with $n+1$ buckets
of size $b$ (total capacity $b(n+1)$) and $bn+d$ keys.
As it is seen in \eqref{Gen-mn2} in $\hat{N}_0(z,w,q)$ the
variable $w$ marks the total capacity, while $z$ marks the
number of keys. So, in this context all the generating
functions to be used for the first two constructions have the form
\begin{equation*}
A(z,w) = 
\sum_{n\geq 0} w^{b(n+1)}.
\sum_{d=0}^{b-1}
A_{bn+d}~\frac{z^{bn+d}}{(bn+d)!}. 
\end{equation*}

To help in fixing ideas, we may think (as an analogue with
equation  \eqref{Gen-mn2}) that
$A_{bn+d} = F_{bn+d}~p_{bn+d}(1)$.

To insert a key, a position is first chosen, and then the
key is inserted. Both actions can be formally specified
using the symbolic method.

\begin{itemize}
\item
\underline{Marking a position}.

Given an almost full table with $n+1$ buckets (the last one 
with $d$ keys, $0 \leq d < b$), a new key can hash in
$n+1$ different places.
As a consequence, we have the counting relation
$C_{bn+d}=(n+1)A_{bn+d}$, leading to the $\partial_w$ relation in their
respective multivariate generating functions. 

Notice that the key has not been inserted yet, only a
position is chosen, and so the total number of keys in the
table does not change.
\item
\underline{Adding a key}. 

Once a position is chosen, then a key is added.
In this setting $C_{bn+d} = A_{bn+d-1}$, leading to the $\int$
relation. No further calculations are needed, since it only
specifies that the number of keys has been increased by 1
(all the other calculations were done when marking the
position).
\end{itemize}

Other constructions are also useful to analyze linear probing
hashing. 

\begin{itemize}

\item
\underline{Bucketing}. 

When considering a single bucket, it has capacity $b$, giving the factor
$w^b$ in the generating function. Moreover, for each $n$, there
is only one way to hash $n$ keys in this bucket (all these
key have this hash position). Since the generating functions
are exponential in $n$, this gives the factor $e^z$ (reflecting
the fact that $C_{m,n} = 1$ for $n\geq 0$ and $m=1$ since there
is only one bucket). In this context $\delta$ is the Dirac's
$\delta$ function ($\delta(a,b)=1$ if $a=b$ and $0$ otherwise).

This construction is used for general hash tables with
$m$ buckets and $n$ keys (not necessarily almost full)
in Sections \ref{SH} and \ref{SRH}.

\item

\underline{Marking a key}.

In some cases, we need to choose a key among $n$
keys that hash to some specific location.
The $q$-analogue of this construction is
used in Section \ref{SRH}. 
The counting relation $C_{m,n} =
nA_{m,n}$ leads to the $\partial_z$ relation in their
respective multivariate generating functions.

\end{itemize}

\subsection{The $q$-calculus.}\label{qanalogues}

In an almost full table with $n+1$ buckets, there are $n+1$
places where a new key can hash. However, if a
distributional analysis is done, its displacement depends on the
place where it has hashed: it is $i$ if the
key hashes to bucket $n+1-i$ with $1 \leq i \leq
n+1$. 
In this context, to keep
track of the distribution of random variables (e.g.\ the displacement
of
a new inserted key), we need generalizations of the constructions above
that belong to the area
of $q$-calculus (equations \eqref{qq2} and \eqref{qq3}).

The same happens to the $\Mark$ construction We rank the $n$ keys 
by the labels $0,\dots,n-1$ (in arbitrary order),
and give the key with label $k$ the weight $q^k$
 (equations \eqref{qq4} and \eqref{qq5}).

We present below some of these translations, where the variable
$q$ marks
the value of the random variable at hand.
Moments result from using the operators $\partial_q$ (differentiation
w.r.t.\ $q$) and $\Uq$ (setting $q=1$).
\begin{align}
n &\mapsto [n]:=1+q+q^2+\ldots+q^{n-1} = \frac{1-q^n}{1-q}, \label{qq1}
\\
\sum_{n\geq 0} (n+1) w^{b(n+1)}&\cdot  
 \sum_{d=0}^{b-1} A_{bn+d}~\frac{z^{bn+d}}{(bn+d)!} \nonumber \\
&\mapsto 
\sum_{n\geq 0} [n+1] w^{b(n+1)}\cdot  \sum_{d=0}^{b-1}
A_{bn+d}(q)~\frac{z^{bn+d}}{(bn+d)!}, \label{qq2} 
\\
\frac{w}{b}\frac{\partial}{\partial w}A(z,w) &\mapsto 
\HH[A(z,w)] :=
\frac{A(z,w)-A(z,wq^{\frac{1}{b}})}{1-q}, \label{qq3} 
\\
\sum_{m\geq 0} w^{m}\cdot  
 \sum_{n\geq 0}n A_{m,n}~\frac{z^{n}}{n!}
&\mapsto 
\sum_{m\geq 0} w^{m}\cdot  \sum_{n\geq 0}
[n] 
A_{m,n}(q)~\frac{z^{n}}{n!}, \label{qq4} 
\\
z \frac{\partial}{\partial z}A(z,w) &\mapsto 
\hat{\HH}[A(z,w)] :=
\frac{A(z,w)-A(qz,w)}{1-q}. \label{qq5} 
\end{align}

\section{Probabilistic method: finite and infinite hash tables}\label{Shash}

In general, consider a hash table, with locations (``buckets'') each
having capacity $b$; we suppose that the buckets are labelled
by $i\in\HT$, for a suitable index set $\HT$.
Let for each bucket $i\in\HT$,
$X_i$ be the number of keys that have hash
address $i$, and thus first try bucket $i$. 
We are mainly interested in the case when the $X_i$ are random, but in this
section $X_i$ can be any (deterministic or random) non-negative integers; we
consider the random case further in the next section.

Moreover, let
$H_i$ be the total number of keys that try bucket $i$ and 
let $Q_i$ be 
the \emph{overflow} from bucket $i$, i.e.,
the number of keys that try bucket $i$ but fail
to find room and thus are transferred to the next bucket.
We call the sequence $H_i$, $i\in \HT$, the \emph{profile} of the hash
table. (We will see that many quantities of interest are determine by the
profile.) 
These definitions yield the equations
\begin{align}
  H_i&=X_i+Q_{i-1},
& 
\label{h}
\\
 Q_i&=(H_i-b)_+. \label{q}
\end{align}
The final number of keys stored in bucket $i$ is
$  Y_i:=H_i\bmin b:=\min(H_i,b)$; 
in particular, the bucket is full if and
only if $H_i\ge b$.

\begin{remark}
The equations \eqrefhq{} are the same as in queuing theory, 
with $Q_i$ the queuing process generated by the random variables $X_i-b$,
see
\cite[Section VI.9, in particular Example (c)]{FellerII}.
\end{remark}

Standard hashing is when
the index set  $\HT$ is the cyclic group $\bbZ_m$.
Another standard case is the {parking problem},
where $\HT$ is an interval $\set{1,\dots,m}$ for some integer $m$; in this
case the $Q_m$ keys that try the last bucket but fail to find room there
are lost (overflow), and \eqrefhq{} use the initial value $Q_0:=0$.

In the probabilistic analysis, we will mainly study infinite hash tables,
either one-sided 
with $\HT=\bbN:=\set{1,2,3,\dots}$, or two-sided with $\HT=\bbZ$; as we
shall see, these occur naturally as limits of finite hash tables.
In the one-sided case, we again define $Q_0:=0$, and then, given
$(X_i)_1^\infty$,  $H_i$ and $Q_i$
are uniquely determined recursively for all $i\ge1$ by 
\eqrefhq{}. 
In the doubly-infinite case, it is not obvious that the equations
\eqrefhq{} 
really have a solution; we return to this question 
in \refL{LH} below.

In the case $\HT=\bbZ_m$, we allow (with a minor abuse of notation) also
the index $i$ in these quantities
to be an arbitrary integer with the obvious interpretation; 
then $X_i$, $H_i$ and so on are periodic sequences defined for
$i\in\bbZ$.

We can express $H_i$ and $Q_i$ in $X_i$ by the following lemma,
which generalizes (and extends to infinite hashing) the case $b=1$ 
treated in \cite[Exercise 6.4-32]{KnuthIII}, \cite[Proposition 5.3]{SJ129},
\cite[Lemma 2.1]{SJ133}.

\begin{lemma}\label{LH}
  Let $X_i$, $i\in\HT$, be given non-negative integers.
  \begin{romenumerate}[5pt]
  \item \label{LH+}
If\/  $\HT=\set{1,\dots,m}$ or  $\bbN$,
then the equations \eqrefhq{}, 
for  $i\in\HT$,
have a unique solution given by, considering $j\ge0$,
\begin{align}
  H_i &=  \max _{j< i} \sum_{k=j+1}^i (X_k-b)+b ,  \label{h1}
\\
  Q_i &=  \max _{j\le i} \sum_{k=j+1}^i (X_k-b) .\label{q1}
\end{align}

\item \label{LHm}
If\/ $\HT=\bbZ_m$, and moreover $n:=\sum_1^m X_i<bm$, 
then the equations \eqrefhq{},
for $i\in\HT$,
have a unique solution given by \eqref{h1}--\eqref{q1}, 
now with  $j\in\bbZ$.
Furthermore, there exists $i_0\in\HT$ such that $H_{i_0}<b$
and thus $Q_{i_0}=0$. 

 \item \label{LHz}
If\/ $\HT=\bbZ$, assume that
\begin{equation}\label{infty}
 \sum_{i=0}^{N-1} (b-X_{-i})\to\infty \qquad \text{as \Ntoo}.  
\end{equation}
Then the equations \eqrefhq{}, 
for $i\in\HT$,
have a solution given by \eqref{h1}--\eqref{q1}, 
with $j\in\bbZ$. This is the minimal solution to \eqrefhq{}, and,
furthermore, for each $i\in\HT$ there exists $i_0<i$ such that $H_{i_0}<b$
and thus $Q_{i_0}=0$. Conversely, 
this is the only solution such that for every $i$ there exists $i_0<i$ with
$Q_{i_0}=0$. 
  \end{romenumerate}
\end{lemma}
In the sequel, we will always use this solution of 
\eqrefhq{} 
for
hashing on $\bbZ$ (assuming that \eqref{infty} holds); we can regard this as
a definition of hashing on $\bbZ$.

Before giving the proof, we introduce the partial sums $S_k$ of $X_i$; 
these are defined by $S_0=0$ and $S_k-S_{k-1}=X_k$,
where for the four cases above we let $S_k$ be defined for
$k\in\set{0,\dots,m}$ when $\HT=\set{1,\dots,m}$,
$k\ge0$ when $\HT=\bbN$, 
$k\in\bbZ$ when $\HT=\bbZ_m$ or $\HT=\bbZ$.
Explicitly, for such $k$, (with an empty sum defined as 0)
\begin{align}
  S_k:=
  \begin{cases}
\sum_{i=1}^k X_i,&  k\ge0,
\\
-\sum_{i=k+1}^0 X_i,&  k<0.	
  \end{cases}
\end{align}
Note that in a finite hash table with $\HT=\bbZ_m$ or $\set{1,\dots,m}$,
the total number $n$ of keys is $S_m$.
(For $\HT=\bbZ_m$, note also that
$S_{k+m}=S_k+n$ for all $k\in\bbZ$, so $S_k$ is not periodic.)

In terms of $S_k$, \eqref{h1}--\eqref{q1} can be written
\begin{align}
  H_i 
&= \max_{j< i} \bigpar{S_i-S_{j}-b(i-j)+b}\label{h2}
\\
&= S_i-bi - \min_{j< i} \xpar{S_j-bj}+b,\label{h3}
\\
  Q_i 
&= \max_{j\le i} \bigpar{S_i-S_{j}-b(i-j)}\label{q2}
\\
&= S_i-bi - \min_{j\le i} \xpar{S_j-bj}.\label{q3}
\end{align}

\begin{remark}
  In the doubly-infinte Poisson model discussed further in \refS{Sconv},
the $X_i$ are \iid{} with $X_i\sim\Po(b\ga)$.
Thus $S_i-bi$ is a random walk with negative drift $\E X_i-b=-b(1-\ga)$.
We can interpret \eqref{h3} and \eqref{q3} as saying that $H_i$ and $Q_i$
are two variants of the corresponding reflected random walk, \ie, this
random walk forced to stay non-negative.
\end{remark}

\begin{proof}[of \refL{LH}]
\pfitemref{LH+}
Here the maxima in \eqref{h1}--\eqref{q1}
are over finite sets (and thus well-defined),
since we consider only  $j\ge0$.
It is clear by induction that the equations
\eqrefhq{} have a unique solution
with $Q_0=0$. Furthermore, \eqrefhq{}  yield
\begin{equation}\label{qi}
  Q_i=(Q_{i-1}+X_i-b)_+ = \max(Q_{i-1}+X_i-b,0)
\end{equation}
and \eqref{q1} follows by induction for all $i\ge0$.
(Note that the term $j=i$ in \eqref{q1}
is $\sum_{i+1}^i(X_k-b)=0$, by definition of an empty sum.) 
Then \eqref{h1} follows by \eqref{h}.

\pfitemref{LHz}
The assumption \eqref{infty} implies that the maxima in
\eqref{h1}--\eqref{q1}  are well defined, since the expressions tend to
$-\infty$ as $j\to-\infty$.
If we define $H_i$ and $Q_i$ by these equations, then 
\eqref{h1}--\eqref{q1} imply $H_i=Q_{i-1}+X_i$, \ie, \eqref{h}.
Furthermore,
\eqref{q1} implies \eqref{qi}, and thus also \eqref{q}. 
Hence $H_i$ and $Q_i$ solve \eqrefhq.
We denote temporarily this solution by $H_i^*$ and $Q_i^*$.

Suppose that $H_i,Q_i$ is any other solution of \eqrefhq{}.
Then
\begin{equation}
  Q_i =(H_i-b)_+ \ge H_i-b =Q_{i-1}+X_i-b
\end{equation}
and thus by induction, for any $j\le i$,
\begin{equation}
  Q_i \ge Q_{j}+ \sum_{k=j+1}^i (X_k-b)
\ge \sum_{k=j+1}^i (X_k-b).
\end{equation}
Taking the  maximum over all $j\le i$ we find $Q_i\ge Q_i^*$,
and thus by \eqref{h} also $H_i\ge H_i^*$.
Hence, $H_i^*,Q_i^*$ is the minimal solution to \eqrefhq{}.

Furthermore, since $S_j-bj\to\infty$ as $j\to-\infty$ by \eqref{infty},
for any $i$ there exists $i_0<i$ such that 
$\min_{j<i_0}(S_j-bj) > S_{i_0}-bi_0$, and hence, by \eqref{h3},
$H^*_{i_0}<b$, which implies $Q^*_{i_0}=0$ by \eqref{q}.

Conversely, if $H_i,Q_i$ is any solution and $Q_{i_0}=0$ for some $i_0\le i$,
let $i_0$ be the largest such index. Then $Q_j>0$ for $i_0<j\le i$, 
and thus by \eqref{q} and \eqref{h},
$Q_j=H_j-b=Q_{j-1}+X_j-b$. Consequently,
\begin{equation}
  Q_i=Q_{i_0}+\sum_{j=i_0+1}^i(X_j-b)
= \sum_{j=i_0+1}^i(X_j-b) \le Q_i^*.
\end{equation}
On the other hand, we have shown that $Q_i\ge Q_i^*$ for any solution. Hence
$Q_i=Q_i^*$. If this holds for all $i$, then also $H_i=H_i^*$ by \eqref{h}.

\pfitemref{LHm}
Solutions of \eqrefhq{} with $i\in\bbZ_m$ can be regarded as periodic
solutions of \eqrefhq{} with $i\in\bbZ$ (with period $m$), with the same
$X_i$. 
The assumption 
$S_m<bm$ implies \eqref{infty}, as is easily seen. (If $N=k+\ell m$, then
$bN+S_{-N}=bk+S_{-k}+\ell (bm-S_m)$.)
Hence we can use \ref{LHz} (for hashing on $\bbZ$) and see that
\eqref{h1}--\eqref{q1} yield a periodic solution to \eqrefhq{}.

Conversely,
suppose that $H_i,Q_i$ is a periodic solution to \eqrefhq{}.
Suppose first that
$H_i\ge b$ for all $i$.
Then, by \eqref{q} and \eqref{h}, $Q_i=H_i-b=Q_{i-1}+X_i-b$, and by induction
\begin{equation*}
Q_m=Q_0+\sum_{j=1}^m(X_j-b)=Q_0+S_m-bm <Q_0,   
\end{equation*}
which contradicts the fact
that $Q_m=Q_0$. 
(This just says that with fewer than $bm$ keys, we cannot fill every
bucket.) 
Consequently, every periodic solution must have $H_i<b$ and $Q_i=0$
for some $i$, and thus by \ref{LHz} the periodic solution is unique.
\end{proof}

\begin{remark}
  In case \ref{LHz}, there exists other solutions, with $Q_i>0$ for all
  $i<-M$ for some $M$. These solutions have 
$H_i\to\infty$ and $Q_i\to\infty$ as $i\to-\infty$.
They correspond to hashing with an
  infinite number of keys entering from $-\infty$, in addition to the
  $X_i$ keys at each finite $i$; these solutions are not interesting for
  our purpose.

In case \ref{LHm}, we have assumed $n=S_m<bm$. There is obviously no solution
if $n>bm$, since then the $n$ keys cannot fit in the $m$ buckets.
If $n=S_m=bm$, so the $n$ keys fit precisely, it is easy to see that
\eqref{h1}--\eqref{q1} still yield a solution; this is the
unique solution with $Q_j=0$ for some $j$. 
(We have $H_i\ge b$, so $Q_i=H_i-b$ and $Y_i=b$ for all $i$.)
There are also other solutions,
giving by adding a positive constant to all $H_i$ and $Q_i$; these
correspond to hashing with some additional keys eternally 
circling around the completely filled hash table, searching in vain for a place;
again these solutions are not interesting.
\end{remark}

\section{Convergence to an infinite hash table}\label{Sconv}

In the exact model, we consider hashing on $Z_m$
with $n$ keys having independent uniformly random hash addresses; thus
$X_1,\dots,X_m$ have a multinomial distribution
with parameters $n$ and $(1/m,\allowbreak\dots,1/m)$. 
We denote these $X_i$ by $X\mni$, and
denote the profile of the resulting random 
hash table by $H\mni$, where as above $i\in Z_m$ but we also can allow $i\in
\bbZ$ in the obvious way. 

We consider a limit with $m,n\to\infty$ and $n/bm\to\ga\in(0,1)$.
The appropriate limit object turns out to be 
an infinite hash table on $\bbZ$ with $X_i=X\gai$ that are 
independent and identically distributed (\iid{}) with the
Poisson distribution $X_i\sim\Po(\ga b)$; this is an infinite version of 
the Poisson model defined in \refS{SPoissonTransform}.
Note that $\E X_i=\ga b<b$, so $\E(b-X_i)>0$ and
\eqref{infty} holds almost 
surely by the law of large numbers; hence this infinite hash table is
well-defined almost surely (a.s.).  
We denote the profile of this hash table by
$H\gai$.

\begin{remark}
We will use subscripts $m,n$ and $\ga$ in the same way for other
random variables too, with $m,n$ signifying the exact model and $\ga$ the
infinite Poisson model. However, we often omit the $\ga$ when it is clear
from the context.
\end{remark}

We claim that the profile $(H\mni)\ioooo$,
regarded as a random element of the product space $\bbZ^\bbZ$, 
converges in distribution to the profile  $(H\gai)\ioooo$.
By the definition of the product topology,
this is equivalent to convergence in distribution of any finite vector
$(H\mni)_{-M}^N$ to $(H\gai)_{-M}^N$.

\begin{lemma}\label{Llim}
  Let $m,n\to\infty$ with $n/bm\to\ga$ for some $\ga$ with $0\le\ga<1$.
Then 
$(H\mni)\ioooo \dto (H\gai)\ioooo$.
\end{lemma}

\begin{proof}
We note first that each $X_{m,n;i}$ has a binomial distribution,
$X_{m,n;i}\sim\Bin(n,1/m)$. 
Since $1/m\to0$ and $n/m\to b\ga$, 
it is well-known that then each $X_{m,n;i}$ is asymptotically Poisson
distributed; more precisely, $X_{m,n;i}\dto \Po(b\ga)$.
Moreover, this extends to the joint distribution for any fixed number of $i$'s;
this well-known fact is easily verified by noting that for every fixed
$M\ge1$ and 
$k_1,\dots,k_M\ge0$, for $m\ge M$ and with $K:=k_1+\dots+k_M$,
\begin{equation*}
  \begin{split}
  \Pr\bigpar{X_{m,n;i}=k_i, i=1,\dots,M}
&
=\binom{n}{k_1,\dots,k_M,n-K} 
\cdot
\prod_{i=1}^M\Bigparfrac{1}{m}^{k_i}
\cdot
\Bigpar{1-\frac{M}{m}}^{n-K}
\\&
=\etto 	
\prod_{i=1}^M\frac{1}{k_i!}\Bigparfrac{n}{m}^{k_i}
\cdot
e^{-M n/m}
\\&
\to 	
\prod_{i=1}^M\frac{(b\ga)^{k_i}}{k_i!}
\cdot
e^{-M b\ga}
=
\prod_{i=1}^M\Pr\bigpar{X_{i;\ga}=k_i}.
  \end{split}
\end{equation*}
Hence, using also the translation invariance, for any $M_1,M_2\ge0$,
\begin{equation}\label{convX}
  \bigpar{X_{m,n;i}}_{i=-M_1}^{M_2}
\dto
  \bigpar{X_{\ga;i}}_{i=-M_1}^{M_2}.
\end{equation}

Next, denote the overflow $Q_i$ for the finite exact model and the infinite
Poisson model by $Q_{m,n;i}$ and $Q_{\ga;i}$, respectively.
We show first that $Q_{m,n;i}\dto Q_{\ga;i}$ for each fixed $i$. By
translation invariance, we may choose $i=0$.

Recall that by \refL{LH}, $Q_0$ is given by \eqref{q1} for both models.
We introduce truncated versions of this: For $L\ge0$, let
\begin{align*}
  Q_{m,n;0}\LL&:=\max_{-L\le j\le 0} \sum_{k=j+1}^0(X_{m,n;k}-b),
&
  Q_{\ga;0}\LL&:=\max_{-L\le j\le 0} \sum_{k=j+1}^0(X_{\ga;k}-b).
\end{align*}
Then \eqref{convX} implies
\begin{equation}\label{erika}
  Q_{m,n;0}\LL\dto Q_{\ga;0}\LL
\end{equation}
for any fixed $L$.

Almost surely \eqref{infty} holds, and thus
$Q_{\ga;0}\LL= Q_{\ga;0}$ for large $L$; hence
\begin{equation}\label{faster}
  \Pr\bigpar{Q_{\ga;0}\neq Q_{\ga;0}\LL}
\to0
\qquad\text{as $L\to\infty$}.
\end{equation}
Moreover,
since $X_{m,n;i}$ is periodic in $i$ with
period $m$, and the sum over a period is $n<bm$, the maximum in \eqref{q1}
is attained for some $j\le0$ with $j>-m$.
Hence, $Q_{m,n;0}=Q_{m,n;0}\Lm$, and furthermore,
$Q_{m,n;0}\neq Q_{m,n;0}\LL$ only if $L<m$ and
the maximum is attained for some
$j\in[-m,\dots,-L-1]$. 
Hence, for any $L\ge0$,
\begin{equation}\label{paddington}
  \Pr\bigpar{Q_{m,n;0}\neq Q_{m,n;0}\LL}
\le 
\sum_{j=-m}^{-L-1} 
  \Pr\lrpar{\sum_{k=j+1}^0(X_{m,n;k}-b)>0}.
\end{equation}
Moreover, for $j\le0$ and $J:=|j|$,
\begin{equation}
  \Pr\lrpar{\sum_{k=j+1}^0(X_{m,n;k}-b)>0}
=
  \Pr\lrpar{\sum_{k=1-J}^{0}X_{m,n;k}>Jb}.
\end{equation}
Here, for $m\ge J$, $\sum_{k=1-J}^{0}X_{m,n;k}\sim \Bin(n,J/m)$,
and a simple Chernoff bound \cite[Remark 2.5]{JLR} yields
\begin{equation}\label{magn}
    \Pr\lrpar{\sum_{k=1-J}^{0}X_{m,n;k}>Jb}
\le \exp\lrpar{-\frac{2 (Jb-Jn/m)^2}{J}}
= \exp\bigpar{-{2J (b-n/m)^2}}.
\end{equation}
By assumption, $b-n/m\to b-b\ga=(1-\ga)b>0$, and thus, 
for sufficiently large $m$,
$b-n/m>0.9(1-\ga)b$ and then
\eqref{paddington}--\eqref{magn} yield
\begin{equation}\label{sw}
  \Pr\lrpar{Q_{m,n;0}\neq Q_{m,n;0}\LL}
\le
\sum_{J=L+1}^m
\exp\bigpar{-J(1-\ga)^2b^2}
\le
\sum_{J=L+1}^\infty
\exp\bigpar{-J(1-\ga)^2b^2}.
\end{equation}

It follows from \eqref{faster} and \eqref{sw} that given any $\eps>0$, we can
find $L$ such that, for all large $m$,
$\Pr\lrpar{Q_{m,n;0}\neq Q_{m,n;0}\LL}<\eps$
and
$\Pr\lrpar{Q_{\ga;0}\neq Q_{\ga;0}\LL}<\eps$. Hence, for any $k\ge0$,
\begin{equation}
\bigabs{\Pr(Q_{m,n;0}=k)-  \Pr(Q_{\ga;0}=k)}
<
\bigabs{\Pr(Q_{m,n;0}\LL=k)-\Pr(Q_{\ga;0}\LL=k)}+2\eps,
\end{equation}
which by \eqref{erika} is $<3\eps$ if $m$ is large enough.
Thus $\Pr(Q_{m,n;0}=k)\to  \Pr(Q_{\ga;0}=k)$ for every $k$, i.e.,
\begin{equation}
  Q_{m,n;0} \dto Q_{\ga;0}.
\end{equation}

Moreover, the argument just given extends to the vector
$(Q_0,X_1,X_2\dots,X_N)$ for any 
$N\ge0$; hence also
\begin{equation}\label{emm}
  \bigpar{Q_{m,n;0},X_{m,n;1},\dots,X_{m,n;N}} 
\dto \bigpar{Q_{\ga;0},X_{\ga;1},\dots,X_{\ga;N}}.
\end{equation}
Since $H_1,\dots,H_N$ by \eqrefhq{} are determined by 
$Q_{0},X_{1},\dots,X_{N}$, \eqref{emm} implies
\begin{equation}
  \bigpar{H_{m,n;1},\dots,H_{m,n;N}} 
\dto \bigpar{H_{\ga;1},\dots,H_{\ga;N}}.
\end{equation}
By translation invariance, this yields also
$(H_{m,n;i})_{i=-M}^N\dto (H_{\ga;i})_{i=-M}^N$
 for any fixed $M$ and $N$, which completes the proof.
\end{proof}

For future use, we note also the following uniform estimates.

\begin{lemma}\label{Lexp}
Suppose that $\ga_1<1$.
Then there exists $C$ and $c>0$ such that $\E e^{c H\mni} \le C$ and
$\E e^{c Q\mni} \le C$ for all $m$ and $n$ with $n/bm\le \ga_1$.
\end{lemma}

\begin{proof}\newcommand\ccga{c}
We may choose $i=m$ by symmetry. 
Let $\ccga:=b(1-\ga_1)>0$, so $b-n/m\ge \ccga$.
By \refL{LH}\ref{LHm} and \eqref{q1},
and a Chernoff bound as in \eqref{magn}, 
for any $x>0$,
\begin{equation*}
  \begin{split}
\Pr(Q_{m,n;m}	\ge x)
&
\le \sum_{j=1}^{m-1} \Pr\lrpar{\sum_{j+1}^m(X_j-b)\ge x}
=\sum_{k=1}^{m-1}\Pr\lrpar{S_k-kb\ge x}
\\&
\le \sum_{k=1}^{m-1}\exp\lrpar{-\frac{2 (kb+x-kn/m)^2}{k}}
\\&
\le \sum_{k=1}^{\infty}\exp\lrpar{-2k(b-n/m)^2-4x(b-n/m)}
\\&
\le \sum_{k=1}^{\infty}\exp\lrpar{-2k\ccga^2-4\ccga x}
=\CC e^{-4\ccga x}.
  \end{split}
\end{equation*}
Hence, $\E e^{\ccga Q_{m,n;m}}\le \CC$. The result for $H_i$ follows because
$H_i\le Q_i+b$ by \eqref{q}.
\end{proof}

Many interesting properties of a hash table are determined by the profile,
and \refL{Llim} then implies limit results in many cases.
We give one explicit general theorem, which apart from applying to 
asymptotic results for several variables
also shows a connection between the combinatorial and probabilistic
approaches.

\begin{theorem}\label{Tconv}
Let  $P$ be a (possibly random) non-negative integer-valued property of a
hash table,  
and assume that (the distribution of) $P$ is determined by the
profile $H_i$ of the hash table. 
Let $0\le\ga<1$ and suppose further that almost surely the profile
$H_{\ga;i}$,
$i\in\bbZ$,
is such that there exists some $N$ such that
(the distribution of) $P$ is the same for every hash
table with a profile that equals $H_{\ga;i}$ for $|i|\le N$.

Let $P_{m,n}$ and $P_\ga$ be the random variables given by $P$ for the exact
model  with $m$ buckets and $n$ keys, and the doubly infinite Poisson model 
with $\HT=\bbZ$ and each $X_i\sim\Po(b\ga)$, respectively; 
furthermore, denote the corresponding probability generating functions by
$p_{m,n}(q)$ and $p_\ga(q)$.

\begin{romenumerate}[5pt]
\item \label{tconvmn}
If\/ $m,n\to\infty$ with $n/bm\to\ga$, then
$P_{m,n}\dto P_\ga$, and thus
$p_{m,n}(q)\to p_\ga(q)$ when $|q|\le1$.

\item \label{tconvPm}
If\/ $m\to\infty$ and $|q|\le1$, then
$\ppm[p_{m,n}(q);b\ga] \to p_{\ga}(q)$.

\item \label{tconvres}
If\/ $0\le\ga<1$ and\/ $|q|\le1$, 
and the 
generating function $P(b\ga,y^{1/b}e^{-\ga},q)$
in \eqref{laPoiss} has a simple 
pole at $y=1$ but otherwise is analytic in a disc
$|y|<r$ with radius $r>1$, then the residue at $y=1$ equals
$-p_\ga(q)$.
\end{romenumerate}
\end{theorem}

\begin{proof}
\pfitemref{tconvmn}
Assume for simplicity that $P$ is competely determined by the profile. (The
random case is similar.)
By the assumptions, $P=f((H_i)_{i=-\infty}^\infty)$  for some function
$f:\bbZ^\bbZ\to \bbZ$, where, moreover, the function $f$ is continuous at
$(H_{\ga;i})_i$ a.s.
Hence, \refL{Llim} and the continuous mapping theorem 
\cite[Theorem 5.1]{Billingsley}
imply
\begin{equation*}
  P_{m,n}=f\bigpar{(H_{m,n;i})_{i}} 
\dto f\bigpar{(H_{\ga;i})_{i}}=P_\ga.
\end{equation*}

\pfitemref{tconvPm}
Fix $q$ with $|q|\le1$.
We can write \eqref{pt} as 
$\ppm[p_{m,n}(q);b\ga]=\E p_{m,N}(q)$,
with $N\sim\Po(b\ga m)$. 
We may assume $N=\sum_{i=1}^m X_i$ 
for a fixed \iid{} sequence $X_i\sim\Po(b\ga)$, and then
$N/m\pto b\ga$ as $m\to\infty$ \as{} by the law of large numbers.
Consequently, by \ref{tconvmn}, 
$p_{m,N}(q)\to p_\ga(q)$ a.s., and thus, 
by dominated convergence using $|p_{m,N}(q)|\le1$, 
\begin{equation*}
  \ppm[p_{m,n}(q);b\ga]=\E p_{m,N}(q)
\to p_\ga(q).
\end{equation*}

\pfitemref{tconvres}
Let the residue be $\rho$. Then, by 
\eqref{laPoiss} and simple singularity analysis
\cite{FlaOdl90,FlaSed09},
$\ppm[p_{m,n}(q);b\ga]\sim -\rho$
as $m\to\infty$, and thus $-\rho=p_\ga(q)$ by \ref{tconvPm}.
\end{proof}

A Boolean property that a hash table either has or has not can be regarded
as a $0/1$-valued property, but in this context it is more natural to
consider the probability that a random hash table has the property.
In this case we obtain the following version of \refT{Tconv}.

\begin{corollary}\label{Cconv}
Let  $P$ be a Boolean property of a
hash table,  
and assume that $P$ satisfies the assumptions of \refT{Tconv}.

Let $p_{m,n}$ and $p_\ga$ be the probabilities that $P$ hold in the exact
model  with $m$ buckets and $n$ keys, and in the doubly infinite Poisson model 
with $\HT=\bbZ$ and each $X_i\sim\Po(b\ga)$, respectively.

\begin{romenumerate}[5pt]
\item \label{cconvmn}
If\/ $m,n\to\infty$ with $n/bm\to\ga$, then
$p_{m,n}\to p_\ga$.

\item \label{cconvPm}
If\/ $m\to\infty$ and, then
$\ppm[p_{m,n};b\ga] \to p_{\ga}$.

\item \label{cconvres}
If\/ $0\le\ga<1$, 
and the bivariate 
generating function $P(b\ga,y^{1/b}e^{-\ga})$
in \eqref{laPoiss2}
has a simple 
pole at $y=1$ but otherwise is analytic in a disc
$|y|<r$ with radius $r>1$, then the residue at $y=1$ equals
$-p_\ga$.
\end{romenumerate}
\end{corollary}
\begin{proof}
 Regard $P$ as a $0/1$-valued property and let $\overline P:=1-P$. 
The results  follow by taking $q=0$ in \refT{Tconv}, applied to $\overline P$,
since
$p_{m,n}=\overline p_{m,n}(0)$ and   
$p_{\ga}=\overline p_{\ga}(0)$.
\end{proof}

\begin{remark}
  The case $n/bm\to\ga=0$ is included above, but rather trivial since
  $X_{\ga;i}=0$ so the limiting infinite hash table is empty.
In the sequel, we omit this case.
\end{remark}

\section{The profile and overflow}\label{SH}
\subsection{\Combinatorial}
Let  $\gO(z,w,q)$ be the generating function for the number of
keys that overflow from 
a hash table (i.e., the number of cars that cannot find a place in the
parking problem)
\begin{equation}\label{Omega}
  \gO(z,w,q) := \sum_{m\ge0}\sum_{n\ge0}\sum_{k\ge0} 
N_{m,n,k} w^{bm} \frac{z^n}{n!} q^k,
\end{equation}
where $N_{m,n,k}$ is the number of hash tables of length $m$ with $n$
keys and overflow $k$.
(We include an empty hash table with $m=n=k=0$ in the sum \eqref{Omega}.)
Thus $w$ marks the number of places in the table, $z$ the number of keys
and $q$ the number of keys that overflow.
The following result has also been presented by
Panholzer \cite{Panholzer:slides} 
and  Seitz \cite{SeitzDiploma}.

\begin{theorem}\label{Tomega}
  \begin{align}\label{omega3}
	\Omega(bz,w,q) 
& = \frac{1}{q^b-w^b e^{bqz}}\cdot
\frac{\prod_{j=0}^{b-1} \left(q - \frac{T(\omega^jzw)}{z}\right)}
{\prod_{j=0}^{b-1} \left(1 - \frac{T(\omega^jzw)}{z}\right)}
.
  \end{align}
\end{theorem}

\begin{proof}
In the empty hash table, there is no overflow, and so $N_{0,0,0}
= 1$.

Let consider now a hash table with $m\ge1$ buckets of size $b$
and its last bucket $m$. The number of keys that probe the
last bucket, are the ones that overflow from bucket $m-1$ 
plus the ones that hash into bucket $m$. All these keys but
the $b$ that stay in the last buckets overflow from this table.

Formally, as a first approach, this can be expressed by
\begin{center}
table $\approx$ empty + table${}*\Bucket(\mathcal{Z})$,
\end{center}
that by means of the constructions presented in Section \ref{q-calculus}
translates into
\begin{equation}
\label{gfovemal}
\Omega(z,w,q) \approx 1 + \Omega(z,w,q) \frac{w^b e^{zq}}{q^b}.
\end{equation}
The variable $q$ in $zq$ marks all the keys that hash
into bucket $m$, and the division by $q^b$
indicates that $b$ keys stay in the last bucket, and as a
consequence do no overflow.

We have to include however, a correction
factor when the total number of 
keys that probe position $m$ is $0\le d < b$. In this case
equation \eqref{gfovemal} gives terms with negative powers
$q^{d-b}$. As a consequence,
\begin{align}
\Omega(z,w,q) = 1 + \Omega(z,w,q)\frac{w^be^{zq}}{q^b} + \sum_{d=0}^{b-1}
(1-q^{d-b}) O_d(z,w),
\label{eqomegabien}
\end{align}
where $O_d(z,w)$ is the generating function for the number of hash
tables that have $d$ keys in bucket $m$.

From equations \eqref{LaOgen}, \eqref{laFd} and \eqref{N0}
we have the following chain of identities:
\begin{align*}
q^b\left(1+ \sum_{d=0}^{b-1} (1-q^{d-b}) O_d(bz,w) \right) 
&=
q^b\left(1+ \frac{N_0(bz,w)}{1-N_0(bz,w)} 
-\frac{N_0(bzq,w/q)}{1-N_0(bz,w)} 
\right) 
\\&=
q^b\frac{1-N_0(bzq,w/q)}{1-N_0(bz,w)} 
\\&=
\frac{q^b{\prod_{j=0}^{b-1} 
\left(1 -\frac{T(\omega^jzw)}{qz}\right)}}
{{\prod_{j=0}^{b-1} 
\left(1 -\frac{T(\omega^jzw)}{z}\right)}}
. 
\end{align*}
Then, the result follows from equation
\eqref{eqomegabien}.
\end{proof}

We can obtain a closed form for the expectation of the overflow from the
generating function in \eqref{omega3}.
Let $Q\mn$ denote the overflow in a random hash table with $m$ buckets
and $n$ keys.
\begin{corollary} \label{CEQ}
\begin{equation}
\E Q\mn =
m^{-n}\sum_{j=0}^n\sum_{k=1}^{\floor{j/b}} \binom{n}j(j-kb)k^{j-1}(m-k)^{n-j}.
\end{equation}
\end{corollary}

\begin{proof}
Taking the derivative at $q=1$ in \eqref{Omega} and \eqref{omega3}, 
we obtain, 
since there are $m^n$ hash tables with $m$ buckets and $n$ keys,
\begin{align}
\sum_{m\ge0}\sum_{n\ge0} \E\,& Q_{m,n} m^n w^{bm}\frac{(bz)^n}{n!}
=
\Uq \partial_q \gO(bz,w,q)
\nonumber
\\
&=-\frac{b-bzw^be^{bz}}{(1-w^be^{bz})^2}
+ \frac{1}{1-w^be^{bz}}\sum_{j=0}^{b-1}\frac{1}{1-T(\go^jzw)/z}
\nonumber
\\
&=b(z-1)\frac{w^be^{bz}}{(1-w^be^{bz})^2}
+ \frac{1}{1-w^be^{bz}}\sum_{j=0}^{b-1}\frac{T(\go^jzw)/z}{1-T(\go^jzw)/z}.
\label{qk1}
\end{align}
We have 
\begin{align}
\frac{1}{1-w^be^{bz}}
&=\summ w^{bm}e^{bmz}
=\summ \sumn m^nw^{bm}\frac{(bz)^n}{n!},
\\
\frac{w^be^{bz}}{(1-w^be^{bz})^2}
&=\summ m w^{bm}e^{bmz}
=\summ \sumn m^{n+1} w^{bm}\frac{(bz)^n}{n!}.
\end{align}
Furthermore, by a well-known application of the Lagrange inversion formula,
for any real $r$,
\begin{equation}
\parfrac{T(z)}{z}^r=e^{rT(z)}= \sumn r(n+r)^{n-1}\frac{z^n}{n!}
\end{equation}
and thus
\begin{equation}\label{LagT}
\frac{T(zw)/z}{1-T(zw)/z}=\sum_{r=1}^\infty w^r\parfrac{T(zw)}{zw}^r
= \sum_{r=1}^\infty\sumn r (n+r)^{n-1}\frac{z^n}{n!} w^{n+r}.
\end{equation}
Substituting $\go^j w$ for $w$ and summing over $j$, we kill all powers of
$w$ that are not multiples of $b$ and we obtain, writing $n+r=kb$,
\begin{equation}\label{LagT1}
\sum_{j=0}^{b-1}\frac{T(\go^jzw)/z}{1-T(\go^jzw)/z}
= b\sum_{k=1}^\infty \sum_{n=0}^{kb-1}(kb-n)(kb)^{n-1}\frac{z^n}{n!} w^{bk}.
\end{equation}
Substituting these expansions in \eqref{qk1} and extracting coefficients we
obtain 
\begin{equation}\label{qll}
m^n\E Q\mn =
nm^n-bm^{n+1}+
\sum_{k=1}^m\sum_{j=0}^n \binom{n}j(kb-j)_+k^{j-1}(m-k)^{n-j},
\end{equation}
where $x_+:=x\etta[x\ge0]$.
Using the the elementary summation
\begin{equation}
\sum_{k=1}^m\sum_{j=0}^n \binom{n}j(j-kb)k^{j-1}(m-k)^{n-j}
=nm^n-bm^{n+1},
\end{equation}
we obtain from \eqref{qll} also
\begin{equation}\label{qll+}
m^n\E Q\mn =
\sum_{k=1}^m\sum_{j=0}^n \binom{n}j(j-kb)_+ k^{j-1}(m-k)^{n-j},
\end{equation}
and the result follows.
\end{proof}

We note the following alternative 
exact formula and asymptotic formula for almost full
tables, both 
taken from
\cite[Theorem 14]{RHBuckets}.
An asymptotic formula when $n/bm\to\ga\in(0,1)$ is given in \refC{CEHQ} below.
\begin{align}\label{eqmn98}
{\E} [Q_{m,n}] &=
\sum_{i\geq 2} \binom{n}{i} \frac{(-1)^{i}}{m^i} \sum_{k=1}^{m} 
k^{i-1} \binom{bk-i}{bk - 1}, \\
 {\E}[Q_{m,bm-1}] &= \frac{\sqrt{2\pi b m}}{4}-\frac{7}{6} +
 \sum_{d=1}^{b-1} \frac{T\left(\go^de^{-1}\right)}{1-T\left(\go^de^{-1}\right)}
 + \frac{1}{48}\sqrt{\frac{2\pi}{bm}}+O\left(\frac{1}{m}\right).
\end{align}

\subsection{\Probabilistic}
For the probabilistic version,
we 
use \refT{Tconv}
and study
in the sequel infinite hashing on $\bbZ$,
with $X_i=X\gai$ 
\iid{} random Poisson variables with  $X_i\sim\Po(\ga b)$, where $0<\ga<1$.
(We consider a fixed $\ga$ and omit it from the notations for convenience.)
Thus $X_i$ has the \pgf{}
\begin{equation}\label{g}
  \psix(\zq):=\E \zq^{X_i} = e^{\ga b(\zq-1)}.
\end{equation}
We begin by finding the distributions of $H_i=H\gai$ and $Q_i=Q\gai$.
Let $\psih(\zq):=\E \zq^{H_i}$ and $\psiq(\zq):=\E \zq^{Q_i}$ denote the
\pgf{s} of 
$H_i$ and $Q_i$ (which obviously do not depend on $i\in\bbZ$), defined at
least for $|\zq|\le1$. 


\begin{theorem}\label{TH}
Let $0<\ga<1$.
The random variables $H\gai$ and $Q\gai$ in the infinite Poisson model have
\pgf{s} $\psih(\zq)$ and $\psiq(\zq)$ that extend to meromorphic functions
given by,
with $\zetal=\zetalx$ as in \eqref{zetal},
  \begin{align}
\psih(\zq)&=
\frac{b(1-\ga)(\zq-1)}
{\zq^b e^{\ga b (1-\zq)}-1} 
\frac
{\prod_{\ell=1}^{b-1}\bigpar{\zq-\zetal}}
{\prod_{\ell=1}^{b-1}\bigpar{1-\zetal}},
\label{psih}
\\
\psiq(\zq)&=
\frac{b(1-\ga)(\zq-1)}{\zq^b- e^{\ga b (\zq-1)}} 
\frac
{\prod_{\ell=1}^{b-1}\bigpar{\zq-\zetal}}
{\prod_{\ell=1}^{b-1}\bigpar{1-\zetal}}.
\label{psiq}
  \end{align}

Moreover, for the exact model,
$H\mni$ and $Q\mni$ converge in distribution to $H\gai$ and\/ $Q\gai$,
respectively, 
as $m,n\to\infty$ with $n/bm\to\ga$;
furthermore, for some $\gd>0$,
their \pgf{s} converge to $\psi_H(q)$ and $\psi_Q(q)$,
uniformly for $|q|\le1+\gd$. 
Hence,
$\E H\mni^\ell\to \E H_\ga^\ell$  and
$\E Q\mni^\ell\to \E Q_\ga^\ell$ 
for any $\ell\ge0$.
\end{theorem}

The formula \eqref{psiq}, which easily implies \eqref{psih}, 
see \eqref{a1} below,
was shown by
the combinatorial method in \cite[Theorem 9]{Viola}.
Indeed, it follows from \refT{Tomega} by \refT{Tconv}\ref{tconvres} 
and \eqref{T'}; we
omit the details.
The formula is also implicit in \cite[Exercise 6.4-55]{KnuthIII}.
We give a  probabilistic proof very similar to the argument in \cite{KnuthIII}.

\begin{proof}
By \eqref{q1}, 
$Q_{i-1}$ depends only on $X_j$ for $j\le i-1$; hence, 
$Q_{i-1}$ is
independent of $X_i$ and thus \eqref{h} yields, for $|\zq|\le1$, using
\eqref{g}, 
\begin{equation}\label{a1}
  \psih(\zq)=\psix(\zq)\psiq(\zq)=e^{\ga b(\zq-1)}\psiq(\zq).
\end{equation}
Furthermore, by \eqref{q}, $Q_i+b = \max(H_i,b)$ and thus, for $|\zq|\le1$,
\begin{equation}\label{a2}
  \begin{split}
\zq^b  \psiq(\zq)-\psih(\zq)
&= \sum_{k=0}^{\infty}\zq^{\max\xpar{k,b}} \PP(H_i=k) 
-\sum_{k=0}^{\infty}\zq^{k} \PP(H_i=k) 
\\&
= \sum_{k=0}^{b-1} \PP(H_i=k) \bigpar{\zq^b-\zq^{k}}
=\pi\xpar{\zq}	
  \end{split}
\end{equation}
for some polynomial $\pi$ of degree $b$.
Combining \eqref{a1} and \eqref{a2} we obtain, 
\begin{equation}\label{a3}
  \bigpar{\zq^b-e^{\ga b(\zq-1)}}\psiq(\zq)=\pi(\zq),
\qquad  |\zq|\le1.
\end{equation}

By \eqref{zetalzq2}, with $q=1$, we have for every $\ell$
\begin{equation}
   \zetal^b= e^{b \ga(\zetal-1)}.
\end{equation}
Substituting this in \eqref{a3} (recalling $|\zetal|\le1$ by \refL{Lzetalzq})
shows that 
$\pi(\zetal)=0$ for every $\ell$.
Since  $\zeta_0,\dots,\zeta_{b-1}$ are distinct, again by \refL{Lzetalzq},
these numbers are 
the $b$ roots of the $b$:th degree polynomial $\pi$, and thus 
\begin{equation}
  \pi(\zq)=\ctt\prod_{\ell=0}^{b-1} (\zq-\zetal)
\end{equation}
for some constant $\ctt$. Using this in \eqref{a3} yields,
recalling $\zeta_0=1$ by \eqref{zeta0}, 
\begin{equation}\label{a5}
  \psiq(\zq)=
\ctt
\frac{\prod_{\ell=0}^{b-1}(\zq-\zetal)}{\zq^b- e^{\ga b (\zq-1)}}
=\ctt
\frac{(\zq-1)\prod_{\ell=1}^{b-1}(\zq-\zetal)}{\zq^b- e^{\ga b (\zq-1)}}.
\end{equation}
To find $\ctt$, we let $\zq\to1$ and use \lhopitals,
which yields
\begin{equation}
1= \psiq(1)=
\ctt
\frac{\prod_{\ell=1}^{b-1}(1-\zetal)}{b- \ga b }.
\end{equation}
Hence, recalling $T_0(b\ga)$ from
\eqref{T00}, see
 \cite[Theorem 7]{Viola} and
\cite[Theorem 4.1]{Bigbuck},
\begin{equation}\label{ctt}
\ctt= 
\frac{b(1-\ga)}{\prod_{\ell=1}^{b-1}(1-\zetal)}
=T_0(b\ga).
\end{equation}
We now obtain \eqref{psiq} from \eqref{a5} and \eqref{ctt};
\eqref{psih} then follows by \eqref{a1}.

The convergence in distribution in the final statement follows from
\refT{Tconv}\ref{tconvmn} (or \refL{Llim}); note that
$H_i$ and $Q_i$ trivially satisfy the condition in \refT{Tconv}. 
The convergence of the \pgf{s} follows from this and Lemma \ref{Lexp} by a
standard argument (for any $\gd<e^c-1$, with $c$ as in \refL{Lexp}).
By another standard result, the convergence of the probability generating
functions in a neighbourhood of $1$ implies convergence of all moments.
\end{proof}

The moments can  be computed from the \pgf{s} \eqref{psih} and 
\eqref{psiq}. We do this explicitly for the expectation only; the formulas for
higher moments are similar but more complicated.
The expectation \eqref{eq} was given in  \cite[Exercise 6.4-55]{KnuthIII}. 
\begin{corollary}\label{CEHQ}
As $m,n\to\infty$ with $n/bm\to\ga\in(0,1)$,
  \begin{align}
\E H\mni &\to \E H_\ga 
= \frac{1}{2(1-\ga)}-\frac{(1-\ga)b}2+\sum_{\ell=1}^{b-1}\frac{1}{1-\zetal},
\label{eh}
\\
\E Q\mni &\to \E Q_\ga 
= \frac{1}{2(1-\ga)}-\frac{(1+\ga)b}2+\sum_{\ell=1}^{b-1}\frac{1}{1-\zetal}.
\label{eq}
  \end{align}
\end{corollary}

\begin{proof}
By the last claim in \refT{TH}, it suffices to compute $\E H_\ga$ and $\E
Q_\ga$. Moreover, in the infinite Poisson model, $\E X_i=b\ga$, and thus
\eqref{h} implies $\E H_\ga=b\ga+\E Q_\ga$. Finally, $\E Q_\ga=\psiq'(1)$
is easily found from \eqref{psiq}, using Taylor expansions in the first
factor.  
\end{proof}

We obtain also results for individual probabilities.
Recall that $Y_i=\min(H_i,b)$ denotes the final number of keys that are
stored in bucket $i$.

\begin{corollary}
  \label{CH}
In the \infpoi{}, for $k=0,\dots,b-1$,
\begin{multline}
\label{ch}  \PP(Y_i=k)=\PP(H_i=k)
=
-b(1-\ga)
\frac{[\zq^k]\prodlb \bigpar{\zq-\zetalx}}
{\prod_{\ell=1}^{b-1}\bigpar{1-\zetalx}}	
\\
=(-1)^{b-k+1}
\frac{b(1-\ga)\ga^{k-b}e_{b-k}\bigpar{\zetaxy0,\dots,\zetaxy{b-1}}}
{\prod_{\ell=1}^{b-1}\bigpar{1-\zetalx}}		  
	\end{multline}
where $e_{b-k}$ is the $(b-k)$:th elementary symmetric function.
In particular,
\begin{equation}\label{ch0}
  \PP(Y_i=0)=\PP(H_i=0)
=(-1)^{b-1}b(1-\ga)
\frac{\prodlbi \zetaly}{\prod_{\ell=1}^{b-1}\bigpar{\ga-\zetaly}}
.
\end{equation}
Furthermore, the probability that a bucket is not full is given by
\begin{equation}\label{ch<b}
  \PP(Y_i<b) = 
  \PP(H_i<b) = 
T_0(b\ga)=
\frac{b(1-\ga)}{\prod_{\ell=1}^{b-1}\bigpar{1-\zetalx}}
\end{equation}
and thus
\begin{equation}
\PP(Y_i=b)=  \PP(H_i\ge b) = 1-T_0(b\ga).
\end{equation}

In the exact model, these results hold asymptotically {\asmn}.
\end{corollary}

\begin{proof}
  By \eqref{psih}, for small $|\zq|$, again using \eqref{zeta0},
  \begin{equation}
\psih(\zq)=-
\frac{b(1-\ga)}{\prod_{\ell=1}^{b-1}(1-\zetal)}
{\prod_{\ell=0}^{b-1}(\zq-\zetal)} + O(|\zq|^b)
  \end{equation}
and \eqref{ch} follows by identifying Taylor coefficients,
recalling  \eqref{zetal}.
Taking $k=0$ we obtain \eqref{ch0}.
Summing \eqref{ch} over $k\le b-1$ yields, using \eqref{ctt},
\begin{equation}
  \begin{split}
\PP(H_i<b)&=
\tba
\sum_{k=0}^{b-1} (-1)^{b-k+1}e_{b-k}(\zeta_0,\dots,\zeta_{b-1})	
\\&
=\tba
\Bigpar{1-\sum_{k=0}^{b}  (-1)^{b-k}e_{b-k}(\zeta_0,\dots,\zeta_{b-1})}
\\&
=\tba
\Bigpar{1-\prodlb(1-\zetal)}
=\tba,
  \end{split}
\end{equation}
since $\zeta_0=1$ by \eqref{zeta0}.
\end{proof}

The generating functions $T_d(u)$ defined in \cite{Viola}
for $0\le d\le b-1$ have the property \cite[p.~318]{Viola}
that $T_d(b\ga)$ is the limit of the probability that in the exact model,
a given bucket
contains more than $d$ empty slots, when $m\to\infty$ and 
$n\sim \Po(\ga bm)$.
This can now be extended by \refC{Cconv}, and we find also the following
relation. 

\begin{theorem} 
In the \infpoi, for $d=0,\dots,b-1$, 
\begin{equation}\label{xa}
T_d(b\ga) =  \PP(Y_i < b-d) = \PP(H_i < b-d)
=\sum_{s=0}^{b-d-1}\PP(Y_i=s).
\end{equation}
Equivalently, for $k=0,\dots,b-1$,
\begin{equation}\label{xb}
  \PP(Y_i=k) = T_{b-k-1}(b\ga)-T_{b-k}(b\ga),
\end{equation}
with $T_{-1}(b\ga):=1$ and $T_{b}(b\ga):=0$.

In the exact model, these results hold asymptotically {\asmn}.
\qed
\end{theorem}

Using \eqref{xb}, it is easy to verify that the formula \eqref{ch} is
equivalent to \cite[Theorem 8]{Viola}.
The results above also yield a simple proof of the following result from
\cite[Theorem 10]{Viola} on 
the asymptotic probability of success in the parking problem,
as $m,n\to\infty$ with $n/m\to\ga$.

\begin{corollary}\label{CQ0}
In the \infpoi,
the probability of no overflow from a given bucket is
  \begin{equation}
	\PP(Q_i=0)
=(-1)^{b-1}b(1-\ga) e^{\ga b} 
\frac{\prodlbi \zetaly}{\prod_{\ell=1}^{b-1}\bigpar{\ga-\zetaly}}
=e^{b\ga} T_{b-1}(b\ga)
.
  \end{equation}

This is the asymptotic probability of success in the parking problem,
as $m,n\to\infty$ with $n/m\to\ga$,
\end{corollary}

\begin{proof}
  Let $\zq=0$ in \eqref{psiq} to obtain the first equality. 
Alternatively, use \eqref{ch0} and 
$\PP(H_i=0)=\PP(Q_i=0)\PP(X_i=0)=\PP(Q_i=0)e^{-b\ga}$ 
from \eqref{h}; this also yields the second equality by \eqref{xa}.
The final claim follows by \refC{Cconv}.
\end{proof}

\section{Robin Hood displacement}\label{SRH}
We follow the ideas presented in \cite{Exact}, \cite{Viola},
\cite{RHBuckets}
and the references therein.
\refF{frh} 
shows the result of inserting keys with the
keys
36, 77, 24, 69, 18, 56, 97, 78, 49, 79, 38 and 10 in a table
with ten buckets
of size 2,
with hash function $h(x)=x$ mod 10, and resolving collisions by
linear
probing using the Robin Hood heuristic.
When there is a collision in bucket $i$ (bucket $i$ is already full),
then the key in this bucket that has probed the least number
of locations,
probes bucket $(i+1)$ mod $m$. In the case of a tie, we
(arbitrarily)
move the key whose key has largest value. 

\begin{figure}[htb]
\noindent\hrule
\medskip
\input{rh.tex}
{\small \caption{ \label{frh} A Robin Hood Linear Probing hash
table.}}
\centerline{}
\medskip
\noindent\hrule
\end{figure}

\refF{rh1} 
shows the partially filled table after
inserting 58. There
is a collision with 18 and 38. Since there is a tie (all of them
are in their first
probe bucket), we arbitrarily decide to move 58, the largest
key.
Then 58 is in its second probe bucket, 78 also, but 49 is in
its first one.
So 49 has to move. Then 49, 69, 79 are all in their second probe
bucket, so
79 has to move to its final position by the tie-break policy
described above.

\begin{figure}[htb]
\noindent\hrule
\medskip
\input{rh1.tex}
{\small \caption{ \label{rh1} The table after inserting 58.}}
\centerline{}
\medskip
\noindent\hrule
\end{figure}


The following properties are easily verified:
\begin{itemize}
\item
At least one key is in its home bucket.
\item
The keys are stored in nondecreasing order by hash value,
starting at some
bucket $k$ and wrapping around. In our example $k$=4
(corresponding to
the home bucket of 24).
\item
If a fixed rule (that depends only on the value of the keys and
not in the
order they are inserted) is used to break ties among the
candidates to probe their
next probe bucket (eg: by sorting these keys in increasing
order), then
the resulting table is independent of the order in which the
keys were
inserted \cite{CelisT}.

As a consequence, the last inserted key
has the same distribution as any other key, and without 
loss of generality we may assume that it hashes to bucket $0$.
\end{itemize}

If we look at a hash table with $m$ buckets 
(numbered $0,\dots,m-1$)
after the first $n$ keys have been
inserted, all the keys that hash to
bucket $0$ (if any) will be occupying contiguous buckets, near
the beginning of the table. The buckets
preceding them will be filled by keys that wrapped around
from the right end of the table, as can
be seen in 
\refF{rh1}. 
The key observation here is that those
keys are exactly the ones that would have
gone to the overflow area. 
Since the displacement $\drh$ of a key $x$
that hashes to 0 is 
the number of buckets before the one containing $x$, and each bucket has
capacity $b$,
it follows that
\begin{equation}\label{drhcrh}
  \drh = \floor{\crh/b},
\end{equation}
where $\crh$, the number of keys that win over $x$ in the
competition for slots in the buckets,
is the sum
\begin{equation}\label{crh}
\crh=Q_{-1}+V
\end{equation}
of the number $Q_{-1}=Q_{m-1}$ of keys that
overflow into $0$ and the number $V$ of keys that hash to 0 that win over
$x$. 
Furthermore, it is easy to see that
the number $Q_{m-1}$ of keys that overflow
does not change when the keys that hash to 0 are removed. 
Hence, 
we may here regard $Q_{-1}=Q_{m-1}$ as the overflow from the hash table
obtained by considering only the buckets $1,\dots,m-1$; this is thus
independent of $V$.

The discussion above assumes $n\le bm$, since otherwise there are no hash
tables with $m$ buckets and $n$ keys. However, for the purpose of defining
the generating functions in \refS{SRHc}, we formally allow any $n\ge0$, taking
\eqref{drhcrh}--\eqref{crh} as a definition when $n>bm$, and then ignoring
bucket 0 and the keys that hash to it when computing $Q_{-1}=Q_{m-1}$.

\subsection{\Combinatorial}\label{SRHc}
We consider the displacement $\drh$ of a marked key $\bullet$.
By symmetry,
it suffices to consider the case when $\bullet$ hashes to the first bucket.
Thus, 
let
\begin{equation}\label{eqRH}
  RH(z,w,q) := \sum_{m\ge0}\sum_{n\ge0}\sum_{k\ge0} 
CRH_{m,n,k} w^{bm} \frac{z^n}{n!} q^k,
\end{equation}
where $CRH_{m,n,k}$ is the number of hash tables of length $m$ with $n$
keys (one of them marked as $\bullet$) such that 
$\bullet$ hashes to the first bucket and
the
displacement $\drh$ of $\bullet$
equals $k$. (I.e., $\bullet$ hashes to bucket 0 but is eventually placed in
bucket $k$.)
Moreover, let $C_{m,n,k}$ be the 
number of hash tables of length $m$ with $n$ keys
keys (one of them marked as $\bullet$) such that 
$\bullet$ hashes to the first bucket and
the variable $\crh$ of $\bullet$ equals $k$,
and let $C(z,w,q)$ be its trivariate generating
function.

In terms of generating function, \eqref{drhcrh} translates to, see
\cite[equation (32)]{Viola},
\begin{equation}\label{complicado}
  \begin{split}
RH(z,w,q) &=  
  \sum_{m\ge0}\sum_{n\ge0}\sum_{k\ge0} 
C_{m,n,k} w^{bm} \frac{z^n}{n!} q^{\floor{\xfrac{k}{b}}} \\
&= \frac{1}{b} \sum_{d=0}^{b-1}
C\left(z,w,\go^dq^{1/b}\right)
\sum_{p=0}^{b-1} \left(\go^dq^{1/b}\right)^{-p}.
	  \end{split}
\end{equation}

\begin{theorem}\label{TRH1}
  \begin{align}\label{trh1d}
	RH(bz,w,q) 
= \frac{1}{b} \sum_{d=0}^{b-1}
C\left(bz,w,\go^dq^{1/b}\right)
\sum_{p=0}^{b-1} \left(\go^dq^{1/b}\right)^{-p},
  \end{align}
with
  \begin{align}\label{trh1c}
C(bz,w,q) 
=\frac{w^b(e^{bz}-e^{bzq})}{(1-q)(q^b-w^be^{bzq})}
\frac{\prod_{j=0}^{b-1} \left(q - \frac{T(\omega^jzw)}{z}\right)}
{\prod_{j=0}^{b-1} \left(1 - \frac{T(\omega^jzw)}{z}\right)}.
  \end{align}
\end{theorem}

\begin{proof}
Equation \eqref{trh1d} has been derived in \eqref{complicado}.
Moreover, in \eqref{crh} $Q_{-1}$ is the overflow already
studied in Section \ref{SH} , so we present here the
combinatorial specification of $V$.

We assume that the marked key $\bullet$
hashes to the first bucket.
Moreover, if $k$ keys collide in the first bucket, then exactly one of
of them has the variable $V$ equal to $i$, for $0\leq i \leq k-1$, 
leading to a contribution of $q^i$ in the generating function.
Then this cost is specified by adding a bucket,
and marking as $\bullet$ an arbitrary key from the ones that hash to it.
As a consequence, we have the specification
\begin{center}
$C$ = Overflow${}* \Mark(\Bucket)$. 
\end{center}
Thus, by 
\eqref{Omega}
and the constructions presented in Section \ref{q-calculus}
(including \eqref{qq5}, the $q$-analogue version of $\Mark$),
\begin{align*}
C(bz,w,q) 
&= \Omega(bz,w,q)  \frac{w^be^{bz}-w^be^{bzq}}{1-q},
\end{align*}  
and the result follows by \refT{Tomega}.
\end{proof}

Moments of the displacement can in principle be found from the 
generating function \eqref{trh1d}. We consider here only the expectation,
for which it is easier to use \refC{CEQ} (or \eqref{eqmn98})
for the overflow together with the
following simple lemma, see 
\cite[6.4-(45) and Exercise 6.4-55]{KnuthIII}.
Note that the expectation of the displacement (but not the variance) is the
same for  any insertion heuristic.
We let $D\mn$ denote the displacement of a random element
in a hash table with $m$ buckets and $n$ keys.
\begin{lemma}\label{LED}
For linear probing with the Robin Hood, FCFS or LCFS (or any other)
heuristic,
\begin{equation}
\E D\mn = \frac{m}n \E Q\mn.
\end{equation}
\end{lemma}
\begin{proof}
For any hash table, and any linear probing
insertion policy, the sum of the $n$ displacements of the keys
equals the sum of the $m$ overflows $Q_i$. Take the expectation.
\end{proof}

We note also (for use in Section \ref{SER})
the following results for the expectation of the displacement for
full tables presented in \cite{TubaT}
and \cite{RHBuckets}.
\begin{align}
b {\E} [D_{m,bm}] &=
\sum_{i\geq 2} \binom{bm-1}{i} \frac{(-1)^i}{m^i} \sum_{k=1}^{m} 
k^{i-1} \binom{bk - i}{bk - 1}
+\frac{m-1}{2m} , \\
b{\E}[D_{m,bm}] &= \frac{\sqrt{2\pi b
m}}{4}-\frac{2}{3}+
\sum_{d=1}^{b-1} \frac{T\left(\go^de^{-1}\right)}{1-T\left(\go^de^{-1}\right)}
+ \frac{1}{48}\sqrt{\frac{2\pi}{bm}}+O\left(\frac{1}{m}\right).
\label{asymptRH}
\end{align}

\subsection{\Probabilistic}

In the infinite Poisson model, it is not formally well defined to talk about
a ``random key''. Instead, we add a new key to the table and consider its
displacement. By symmetry, we may assume that the new key hashes to 0, and
then its displacement $\drh_\ga$ is given by \eqref{drhcrh}--\eqref{crh},
with $Q_{-1}=Q_{\ga;-1}$ and $V=V_\ga$ independent. Furthermore, $Q_{-1}$
has the probability generating function \eqref{psiq} and given $X_0$, $V$ is
a uniformly random integer in \set{0,\dots,X_0}. 

Similarly, in the exact model, by the discussion above we may study the
Robin Hood displacement of the last key instead of taking a random key;
this is the same as the displacement of a new key added to a table with
$m$ buckets and $n-1$ keys.

\begin{theorem}\label{TRH}
  Let $0<\ga<1$.
In the infinite Poisson model,
the variable $V_\ga$, 
the number of keys that win over the new key $\crh_\ga$ and its
Robin Hood displacement $\drh_\ga$ have the \pgf{s} 
\begin{align}
  \psiv(q)
&=\frac{1-e^{b\ga (q-1)}}{b\ga(1-q)}
\label{psiv}
\\
\psicrh(q)&=\psiq(q)\psiv(q)=
\frac{1-\ga}{\ga}\,
\frac{1-e^{b\ga(q-1)}}{e^{b\ga  (q-1)}-q^b} 
\,
\frac
{\prod_{\ell=1}^{b-1}\bigpar{q-\zetal}}
{\prod_{\ell=1}^{b-1}\bigpar{1-\zetal}}.
\label{psic}
\\
\psirh(q)
&=\frac1b\sum_{j=0}^{b-1} \psicrh \bigpar{\go^j q^{1/b}}
\frac{1-q\qw}{1-\go^{-j}q^{-1/b}}
.
\label{psirh}
\end{align}

Moreover, for the exact model,
as $m,n\to\infty$ with $n/bm\to\ga$,
$V\mn\dto V_\ga$,
$\crh\mn\dto \crh_\ga$ and
$\drh\mn\dto \drh_\ga$,
with convergence of all moments;
furthermore, 
for some $\gd>0$, the corresponding \pgf{s} converge,
uniformly for $|q|\le1+\gd$.
\end{theorem}

\begin{proof}
For the convergence, we consider the displacement of a new key added to
the hash table. By symmetry we may assume that the new key hashes to 0,
and then \eqref{drhcrh}--\eqref{crh} apply, 
where $Q_{-1}$ by  \eqref{q} is determined by 
$H_{-1}$ and $V$ is random, given the hash table, with a
distribution determined by $X_0$, and thus by $H_0$ and $Q_{-1}$, and thus by
$H_{-1}$ and $H_0$. Consequently, \refT{Tconv} applies, and yields the
convergence in distribution. (Since we consider adding a new key to the
table, 
this really proves
$\drh_{m,n+1}\dto\drh_\ga$, etc.; we may obviously replace $n$ by $n-1$ in
order to get the result.)

For the \pgf{s} and moments, note that for the exact model, 
for every $c>0$,
\begin{equation}
\E (1+c)^{ V_{m,n}}
\le \E (1+c)^{X_0} = \biggpar{1+\frac{c}m}^n \le e^{c n/m} \le e^{bc}.
\end{equation}
This together with \eqref{crh}, \refL{Lexp} and \Holder's
inequality yields, for some $c_1>0$ and  $C_1<\infty$,
\begin{equation}
  \E e^{c_1\crh_{m,n}}
\le \lrpar{\E e^{2c_1 V_{m,n}} \E e^{2c_1Q_{m,n}}}\qq \le C_1.
\end{equation}
The convergence of \pgf{s} and moments now follows from the convergence in
distribution, using $0\le\drh\le\crh$.

For the distributions for the Poisson model, note that
if $X_0=k\ge0$, then there are, together with the new key, $k+1$ keys
competing at 0, and the number $V=V_\ga$ of them that wins over the new key is
uniform on \set{0,\dots,k}. Thus
  \begin{equation}
	\begin{split}
	  \E (q^V\mid X_0=k)
= \frac{1+\dots+q^k}{k+1}
=\frac{1-q^{k+1}}{(k+1)(1-q)}.
	\end{split}
  \end{equation}
Hence, since $X_0\sim \Po(b\ga)$,
  \begin{equation*}
	\begin{split}
	  \E (q^V)
=\sumk \frac{(b\ga)^k}{k!}e^{-b\ga}\frac{1-q^{k+1}}{(k+1)(1-q)}
=\frac{e^{-b\ga}}{b\ga(1-q)}
\sumk \frac{(b\ga)^{k+1}-(b\ga q)^{k+1}}{(k+1)!},
	\end{split}
  \end{equation*}
yielding \eqref{psiv}. 
This, \eqref{crh} and \eqref{psiq} yields \eqref{psic}.
Finally, \eqref{drhcrh} then yields \eqref{psirh}, 
\cf{} \cite{Viola},
\end{proof}

\begin{corollary}\label{CRH}
As $m,n\to\infty$ with $n/bm\to\ga\in(0,1)$,
\begin{align}\label{ec}
\E \crh_{m,n} &\to \E\crh_\ga 
= \frac{1}{2(1-\ga)}-\frac{b}2+\sum_{\ell=1}^{b-1}\frac{1}{1-\zetal},
\\
\E \drh_{m,n} &\to \E\drh_\ga 
= \frac{1}{2b\ga}\biggpar{\frac{1}{1-\ga}-b-b\ga}
+\frac{1}{b\ga}\sum_{\ell=1}^{b-1}\frac{1}{1-\zetal}.
\label{ed}
\end{align}
\end{corollary}

\begin{proof}
In the \infpoi,  by symmetry, $\E (V_\ga\mid X_0)=\frac12X_0$, and thus 
$\E V_\ga=\frac12\E X_{0}=\frac12b\ga$.
Consequently, by \eqref{crh},
\begin{equation}
\E\crh_\ga=\E Q_\ga+\E V_\ga
=\E Q_\ga+\tfrac12b\ga,  
\end{equation}
which yields \eqref{ec} by \eqref{eq}.

For $\E \drh_\ga$ we differentiate \eqref{psirh}, 
for $j=0$ using $(1-q\qw)/(1-q^{-1/b})=1+q^{-1/b}+\dots+q^{-(b-1)/b}$,
and obtain
\begin{equation}\label{ahlin}
  \begin{split}
\E \drh_\ga = \psirh'(1)
=
\frac{1}b\psicrh'(1) 
- \frac{1}b\psicrh(1)\frac{b-1}2
+\frac1b\sum_{j=1}^{b-1} \psicrh \bigpar{\go^j}
\frac{1}{1-\go^{-j}}
  \end{split}
\end{equation}
where $\psicrh'(1)=\E\crh_\ga$ is given by \eqref{ec} and $\psicrh(1)=1$.
We compute the sum in \eqref{ahlin} as follows. (See the proof of
\cite[Theorem 14]{Viola} for an alternative method.) 
By \eqref{psic},
$\psicrh(\go^j)=-\frac{1-\ga}\ga p(\go^j)$ where 
$p(q):=\xfrac
{\prod_{\ell=1}^{b-1}\bigpar{q-\zetal}}
{\prod_{\ell=1}^{b-1}\bigpar{1-\zetal}}$ is a polynomial of degree $b-1$. 
Define 
\begin{equation}
f(q):=\frac{p(q)}{(q^b-1)(q-1)};
\end{equation}
then $f$ is a rational function with
poles at $\go^j$, $j=0,\dots,b-1$. Furthermore, $f(q)=O(|q|\qww)$ as
$|q|\to\infty$, so integrating $f(z)\dd z$ around the circle $|q|=R$ and
letting $R\to\infty$, the integral tends to 0 and by Cauchy's residue
theorem, the sum of the residues of $f$ is 0. The residue of $f$ at $q=\go^j$,
$j=1,\dots,b-1$ is
\begin{equation}
\frac{p(q)}{bq^{b-1}(q-1)} 
=
  \frac{p(\go^j)}{b(1-\go^{-j})}
\end{equation}
and the residue at the double pole $q=1$ is, after a simple calculation,
\begin{equation}
  -\frac{b-1}{2b}p(1)+\frac{1}b p'(1).
\end{equation}
Consequently,
\begin{equation}\label{talman}
  \begin{split}
\frac{1}b \sum_{j=1}^{b-1} \frac{\psicrh(\go^j)}{1-\go^{-j}}
&=
  -\frac{1-\ga}{\ga }\sum_{j=1}^{b-1} \frac{p(\go^j)}{b(1-\go^{-j})}
=\frac{1-\ga}{\ga}\biggpar{-\frac{b-1}{2b}p(1)+\frac{1}{b} p'(1)}
\\&
=\frac{1-\ga}{\ga b}\biggpar{-\frac{b-1}{2}
+ \sum_{j=1}^{b-1}\frac{1}{1-\zetal}}.
  \end{split}
\raisetag{1.5\baselineskip}
\end{equation}
We obtain \eqref{ed} by combining \eqref{ahlin},
\eqref{ec} and \eqref{talman} 
\end{proof}
\begin{remark}
The result presented in \eqref{ed} may also be directly derived from 
\refC{CEHQ} and \refL{LED}.
\end{remark}
\begin{remark}
The probabilities $\PP(\drh=k)$ can be obtained by extracting the
coefficients of $\psi_C$ (Theorem 13 in \cite{Viola}).
\end{remark}

\section{Block length}\label{SB}

We want to consider a ``random block''.
Some care has to be taken when defining this; for example, (as is well-known)
the block containing a given bucket is \emph{not} what we want.
(This would give a size-biased distribution, see \refT{TBH}.) 
We do this slightly differently in the combinatorial and probabilistic
approaches. 

\subsection{\Combinatorial}\label{SBC}
By symmetry, it suffices as usual to consider hash tables such that the
rightmost bucket is not full, and thus ends a block; we consider that block.
Thus, let
\begin{equation}\label{eqBL}
  B(z,w,q) := \sum_{m\ge0}\sum_{n\ge0}\sum_{k\ge0} 
B_{m,n,k} w^{bm} \frac{z^n}{n!} q^k,
\end{equation}
where $B_{m,n,k}$ is the number of hash tables with $m$ buckets and
$n$ keys such that 
the rightmost bucket is not full
and the last block has length $k$.

\begin{theorem}\label{TRBL}
  \begin{align}\label{trbl}
B(bz,w,q) = 
\frac{1-\prod_{j=0}^{b-1} \left(1 - \frac{T(\omega^jzwq^{1/b})}{z}\right)}
{\prod_{j=0}^{b-1} \left(1 - \frac{T(\omega^jzw)}{z}\right)}.
\end{align}
\end{theorem}

\begin{proof}
In an almost full table the length of the block is marked by $w^b$
in $N_0(bz,w)$. Then, in the combinatorial model,
the generating function $B(z,w,q)$ for the block length is,
using \eqref{N0} and \eqref{laLambda0},
\begin{align*}
B(bz,w,q) = \Lambda_0(bz,w) N_0(bz,wq^{1/b}) =
\frac{1-\prod_{j=0}^{b-1} \left(1 - \frac{T(\omega^jzwq^{1/b})}{z}\right)}
{\prod_{j=0}^{b-1} \left(1 - \frac{T(\omega^jzw)}{z}\right)}.
\end{align*}
\end{proof}

Let $B\mn$ be the length of a random block, chosen uniformly among all
blocks in all hash tables with $m$ buckets and $n$ keys.
This is the same as the length of the last block in a uniformly random
hash table such that the rightmost bucket is not full.
Recall that we denote the number of such hash tables by $Q_{m,n,0}$.

\begin{corollary}\label{CEB}
If\/ $0\le n< bm$, then  
\begin{equation}\label{ceb}
\E B\mn = \frac{m^n}{Q_{m,n,0}}.
\end{equation}
\end{corollary}
\begin{proof}
This can be shown 
by taking the derivative at $q=1$ in \eqref{trbl} after some manipulations
similar to \eqref{LagT1}.
However, it is simpler to note that the sum of the block lengths in any hash
table is $m$, and thus the sum of the lengths of all blocks in all tables is
$m\cdot m^n$, while the number of blocks ending with a given bucket is
$Q_{m,n,0}$ and thus the total number of blocks is
$m\cdot Q_{m,n,0}$.
\end{proof}

\subsection{\Probabilistic}
For the probabilistic version, 
we consider one-sided infinite hashing on $\HT=\bbN$, with
$X_i\sim\Po(\ga b)$ \iid{} as above, and
let $B$ be the length of the first block, i.e.,
\begin{equation}
B=B_\ga:=\min\set{i\ge1:Y_i<b}
=\min\set{i\ge1:H_i<b}.  
\end{equation}

\begin{remark}\label{RB}
We consider here the first block in hashing on $\bbN$. 
Furthermore, since by definition
there is no overflow from  a block, the second block, the third block, and
so on all have the same distribution. 

Moreover, for our usual \infpoi{} on $\bbZ$, 
it is easy to see, using the independence of the $X_i$'s, that we obtain the
same distribution if we for any fixed $i$ 
condition on $H_i<b$, (\ie, on that a block ends at $i$), 
and then take the length of the block starting at $i+1$.

We also obtain the same distribution if
we fix any
$i$ and consider the length of the first (or second, \dots) block
\emph{after} $i$, or similarly the last block before $i$.
\end{remark}

Hence, $B$ is the first positive index $i$ such that the number of keys
$S_i=X_1+\dots+X_i$ hashed to the $i$ first buckets
is less than the capacity $bi$ of these buckets, \ie,
\begin{equation}\label{b2}
B = \min\set{i\ge1:S_i < bi}.
\end{equation}
(This also follows from \refL{LH}.)
In other words, if we consider the random walk
\begin{equation}
S_n':=S_n-bn =\sumin (X_i-b),
\end{equation}
the block length $B$ is the first time this random walk becomes negative.
Since $\E(X_i-b)=\ga b-b<0$, it follows from the law of large numbers that
\as{} $S_n'\to-\infty$ as $\ntoo$, and thus $B<\infty$.

Note also that $S_{B-1}'\ge0$, and thus $0>S'_B\ge -b$.
In fact, the number of keys that hash to the first $B$ buckets is
$S_B=S'_B+bB$, and since all buckets before $B$ are full and thus take
$(B-1)b$ keys, the number of keys in the final bucket of the block
is
\begin{equation}\label{yb}
Y_B=H_B=S_B-(B-1)b=S'_B+b \in\set{0,\dots,b-1}.
\end{equation}

Recall the numbers $\zetal(\zq) =\zetal(\zq;\ga)$ defined in
\eqref{zetalzq}.

\begin{theorem}\label{TB}
Let $0<\ga<1$. 
The \pgf{} $\psib(\zq):=\E \zq^{B}$ of $B=B_\ga$ is given by
\begin{equation}
\label{psib}
\psib(\zq) 
= 1-\prodlb \bigpar{1-\zetal(\zq)}
,
  \end{equation}
for $|\zq|\le R$ for some $R>1$, 
where
$  \zetal(\zq) 
=  \zetal(\zq;\ga) 
$
is given by \eqref{zetalzq}.

More generally, for $|\zq|\le R$ and $t\in\bbC$,
\begin{equation}\label{psiby}
\E \bigpar{\zq^B t^{Y_B}}
=
\E \bigpar{\zq^B t^{H_B}}
=t^b-\prodlb \bigpar{t-\zetal(\zq)}.
\end{equation}
\end{theorem}

\begin{remark}
Related results in a case where $X_i$ are bounded but otherwise have an
arbitrary distribution are proved by a similar but somewhat different
method in  
\cite[Example XII.4(c) and Problem XII.10.13]{FellerII}.
\end{remark}

\begin{proof}
We consider separately the different possibilities for $S'_B$ (or
equivalently, see \eqref{yb},
for the number of keys in the final bucket of the block)
and define $b$ partial \pgf{s} $f_1(\zq),\dots,f_b(\zq)$ by 
  \begin{equation}\label{fk}
	f_k(\zq) := \E \bigpar{\zq^B \och{S'_B=-k}} 
= \sum_{n=1}^\infty \PP\bigpar{B=n \text{ and } S'_B=-k}\zq^n,
\quad k=1,\dots,b.
  \end{equation}

Let $\zeta$ and $\zq$ be non-zero
complex numbers with $|\zeta|,|\zq|\le1$, and define the complex random variables
\begin{equation}\label{ql}
Z_n := \zeta^{S'_n} \zq^n = \prodin (\zeta^{X_i-b}\zq),
\qquad n\ge0.
\end{equation}
We have, by \eqref{g},
\begin{equation}\label{qk}
  \E\bigpar{ \zeta^{X_i-b}\zq}
=\zeta^{-b}\zq \psix(\zeta)
=\zeta^{-b}\zq e^{\ga b(\zeta-1)}.
\end{equation}
Fix an arbitrary $\zq$ with $0<|\zq|\le1$ and
choose $\zeta=\zetal(\zq)$, noting $|\zeta|\le1$ by \refL{Lzetalzq}.
Then \eqref{zetalzq2} holds, and thus \eqref{qk}
reduces to  
$  \E\bigpar{ \zeta^{X_i-b}\zq}=1$. Since the random variables $X_i$ are
independent, it then follows from
\eqref{ql} that $(Z_n)_{n=0}^\infty$ is a martingale.
(See \eg{} \cite[Chapter 10]{Gut} for basic martingale theory.)
Moreover, \eqref{b2} shows that 
$B$ is a stopping time for the corresponding sequence of $\gs$-fields.
Hence also the stopped process 
$(Z_{n\bmin B})_{n=0}^\infty$ is a martingale.
Furthermore, for $n\le B$, we have $S'_n\ge -b$ and thus by \eqref{ql}
$|Z_n|\le |\zeta|^{-b}$, so the martingale
$(Z_{n\bmin B})_{n=0}^\infty$ is bounded.
Consequently, 
by a standard martingale result (see \eg{} \cite[Theorem 10.12.1]{Gut})
together with
 \eqref{fk},
\begin{equation}
  1=\E Z_0 = \E \lim_\ntoo Z_{n\bmin B} = \E Z_B = \E \zeta^{S_B'} \zq^B
=\sum_{k=1}^b \zeta^{-k} f_k(\zq).
\end{equation}
Thus, for any $\zq$ with $0<|\zq|\le1$, the $b$ different choices
$\zeta=\zetal(\zq)$ 
yield $b$ linear equations
\begin{equation}\label{fkb}
\sum_{k=1}^b \zetal(\zq)^{-k} f_k(\zq) =1,
\qquad \ell=0,\dots,b-1
\end{equation}
in the $b$ unknowns $f_1(\zq),\dots,f_b(\zq)$. 
Note that the coefficient matrix of this system of equations is
(essentially) a Vandermonde matrix, and since
$\zeta_0(\zq),\allowbreak \dots, \zeta_{b-1}(\zq)$ are distinct, 
its determinant is non-zero, 
so the system of equations \eqref{fkb} has a unique solution.

To find the solution explicitly,
let us temporarily define $f_0(\zq):=-1$. Then \eqref{fkb} can be written
\begin{equation}\label{fkbb}
\sum_{k=0}^b \zetal(\zq)^{-k} f_k(\zq) =0.
\end{equation}
 Define the polynomial
\begin{equation}\label{qm}
  p(t):=\sum_{k=0}^b f_k(\zq) t^{b-k};
\end{equation}
then by \eqref{fkbb}, $p(\zetal(\zq))=0$ and thus $p(t)$ has
the $b$ (distinct) roots 
$\zeta_0(\zq),\allowbreak\dots,\allowbreak\zeta_{b-1}(\zq)$; 
since further $p(t)$ has leading term $f_0(\zq)t^b=-t^b$, this implies
\begin{equation}\label{qj}
  p(t)=-\prodlb\bigpar{t-\zetal(\zq)}.
\end{equation}
Furthermore, by \eqref{yb} and \eqref{fk}, for any $t\in\bbC$,
\begin{equation}
  \E\bigpar{\zq^Bt^{Y_B}}
= \sum_{k=1}^b \E\bigpar{\zq^B \och{S'_B=-k}} t^{b-k}
= \sum_{k=1}^b f_k(\zq) t^{b-k}. 
\end{equation}
Hence, by \eqref{qm} and \eqref{qj},
\begin{equation}
  \E\bigpar{\zq^Bt^{Y_B}}
= p(t) -f_0(\zq) t^b
= p(t) + t^b
= t^b -\prodlb\bigpar{t-\zetal(\zq)}.
\end{equation}

This proves \eqref{psiby} for $0<|\zq|\le1$, and 
\eqref{psib} follows by taking $t=1$;
the results extend by
analyticity and continuity to $|\zq|\le R$.
\end{proof}

\begin{remark}
By identifying coefficients in \eqref{qj} (or \eqref{psiby}) we also obtain 
\begin{equation}
  f_k(\zq)=(-1)^{k-1} e_k\bigpar{\zeta_0(\zq),\dots,\zeta_{b-1}(\zq)},
\end{equation}
which gives the distribution of the length of blocks with a given number of
keys in the last bucket. In particular, taking $\zq=1$ and using
\eqref{yb} and \eqref{fk},
\begin{equation}
\PP(Y_B=b-k)=
\PP(S'_B=-k)=
  f_k(1)=(-1)^{k-1} e_k\bigpar{\zeta_0,\dots,\zeta_{b-1}}.
\end{equation}
By \eqref{ch} and \eqref{ch<b}, this says that $Y_B$ has the same
distribution as $(Y_i\mid Y_i<b)$, the number of keys placed in a
fixed bucket in hashing on $\bbZ$, conditioned on the bucket not being full.
This is (more or less) obvious, since the buckets that are not full are
exactly the last buckets in the blocks.
\end{remark}

\begin{corollary}\label{CB}
The random block length $B=B_\ga$ defined above has expectation
\begin{equation}\label{eb}
  \E B_\ga = \frac{1}{\tba}
\end{equation}
and variance, with $\zetal$ given by \eqref{zetal},
\begin{equation}\label{vb}
  \Var B_\ga = \frac{1}{b(1-\ga)^2T_0(b\ga)}
-  \frac{2}{b T_0(b\ga)}
 \sum_{\ell=1}^{b-1}\frac{\zetal}{(1-\zetal)(1-\ga\zetal)}
-\frac{1}{T_0(b\ga)^2}.
\end{equation}
\end{corollary}

\begin{proof}
We have from \eqref{psib}, recalling that $\zeta_0(1)=1$ and using \eqref{ctt},
  \begin{equation}
	\E B_\ga = \psib'(1)=\zeta_0'(1) \prodlbi\bigpar{1-\zetal(1)}
=\zeta_0'(1) \frac{b(1-\ga)}{\tba}.
  \end{equation}
Furthermore, (logarithmic) differentiation of \eqref{zetalzq} yields,
using \eqref{T'},
\begin{equation}\label{zetal'}
  \zetal'(\zq) = \frac{\zetal(\zq)}{b\zq(1-\ga\zetal(\zq))},
\end{equation}
and in particular
\begin{equation}\label{zeta0'1}
  \zeta_0'(1) = \frac{1}{b(1-\ga)},
\end{equation}
and \eqref{eb} follows.

Similarly,
  \begin{equation}
	\E B_\ga (B_\ga-1) = \psib''(1)
=\zeta_0''(1) \prodlbi\bigpar{1-\zetal(1)}
-2\sum_{j=1}^{b-1}\zeta_0'(1)\zeta_j'(1)
\frac{\prodlbi\bigpar{1-\zetal(1)}}{1-\zeta_j(1)}
  \end{equation}
and \eqref{vb} follows after a calculation, using
\eqref{ctt}, \eqref{zetal'}--\eqref{zeta0'1} and, by \eqref{zetal'},
differentiation and \eqref{zetal'} again,
  \begin{equation}\label{zetal''}
	q\zetal''(q)+\zetal'(q)=\bigpar{q\zetal'(q)}'
=\frac{\zetal'(q)}{b(1-\ga\zetal(q))^2}
=\frac{\zetal(q)}{b^2q(1-\ga\zetal(q))^3}
.  \end{equation}
We omit the details.
\end{proof}

As said above, the length of the block $\hB_i$
containing a given bucket $i$ has a
different, size-biased distribution. We consider both the exact model and
the \infpoi, and use the notations $\hB_{m,n}$ and $\hB_\ga$ in our usual way.
We first note an analogue of \refL{Lexp}.

\newcommand\ccB{c_1}
\begin{lemma}\label{LexpB}
Suppose that $\ga_1<1$.
Then there exists $C$ and $\ccB>0$ such that $\E e^{\ccB \hB\mn} \le C$ 
for all $m$ and $n$ with $n/bm\le \ga_1$.
\end{lemma}

\begin{proof}\newcommand\ccga{c}
Let again $\ccga:=b(1-\ga_1)>0$, so $b-n/m\ge \ccga$.
Consider the block containing bucket 0. If this block has length $k\ge2$
and starts at $j$, then 
$-k+1\le j\le 0$ (modulo $m$)
and $\sum_{i=j}^{j+k-2}(X_i-b)\ge0$. 
Hence, 
by a Chernoff bound as in \eqref{magn}, 
\begin{equation*}
  \begin{split}
\Pr(\hB_{m,n}=k)
&
\le  k \Pr\bigpar{S_{k-1}-(k-1)b\ge 0}
\le k\exp\bigpar{-2(k-1)\ccga^2}.
  \end{split}
\end{equation*}
The result follows with $\ccB=\ccga^2$.
\end{proof}

\begin{theorem}\label{TBH}
In the \infpoi,  
$\hB=\hB_\ga$ has the size-biased distribution
\begin{equation}\label{tbh}
  \Pr(\hB_\ga=k)=\frac{k\Pr(B_\ga=k)}{\E B_\ga}=T_0(b\ga)k\Pr(B_\ga=k)
\end{equation}
and thus the \pgf
\begin{equation}\label{psihb}
  \psihb(q) = T_0(b\ga) q \psib'(q)
=T_0(b\ga)q\sum_{\ell=0}^{b-1}\zetal'(q)\prod_{j\neq \ell}(1-\zeta_j(q)).
\end{equation}

Moreover, for the exact model,
as $m,n\to\infty$ with $n/bm\to\ga$,
$\hB\mn\dto \hB_\ga$
with convergence of all moments;
furthermore, 
for some $\gd>0$, the \pgf{} converges to $\psihb(q)$,
uniformly for $|q|\le1+\gd$.
\end{theorem}
\begin{proof}
  Consider the block containing bucket 0. This block has length $k$ if and
  only if there is some $i<0$ with $i\ge-k$ such that $H_i<b$ and the block
  starting at $i+1$ has length $k$ (and thus contains 0). 
Hence, using \refR{RB} and \eqref{ch<b},
\begin{equation}\label{hba}
  \begin{split}
\Pr(\hB=k)=\sum_{i=-k}^{-1}	\Pr(H_{-i}<b)\Pr(B=k)
=k T_0(b\ga) \Pr(B=k),
  \end{split}
\end{equation}
which together with \eqref{eb} shows \eqref{tbh}.
(Note also that summing \eqref{hba} over $k$ yields another proof of
\eqref{eb}.) 
The formulas \eqref{psihb} for the \pgf{} follow from \eqref{tbh} and
\eqref{psib}. 

The convergence in distribution of $\hB\mn$ follows from \refT{Tconv}.
Convergence of moments and \pgf{} then follows using \refL{LexpB}.
\end{proof}
  
\begin{corollary}\label{CBH}
As $m,n\to\infty$ with $n/bm\to\ga\in(0,1)$,
\begin{equation}\label{ebh}
\E\hB\mn\to\E\hB_\ga 
= \frac{1}{b(1-\ga)^2}
-  \frac{2}{b}
 \sum_{\ell=1}^{b-1}\frac{\zetal}{(1-\zetal)(1-\ga\zetal)}
.
\end{equation}
\end{corollary}

\begin{proof}
  The convergence follows by \refT{TBH}, which also implies
\begin{equation}\label{ebhb}
\E\hB_\ga = \sum_k k \Pr(\hB_\ga=k)
=\sum_k \frac{k^2\Pr(B_\ga=k)}{\E B_\ga}
= \frac{\E B_\ga^2}{\E B_\ga}.
\end{equation}
The result follows by \refC{CB}.
\end{proof}

Recall $B\mn$ defined in \refS{SBC}.
\begin{theorem}\label{TBHmn}
In the exact model,
$\hB\mn$ has the size-biased distribution
\begin{equation}\label{tbhmn}
  \Pr(\hB\mn=k)=\frac{k\Pr(B\mn=k)}{\E B\mn}=\frac{Q_{m,n,0}}{m^n}k\Pr(B\mn=k).
\end{equation}
\end{theorem}

\begin{proof}
$\hB\mn$ is defined as the length of the block containing a given bucket
$i$. We may instead let $i$ be a random bucket, which means that among all
blocks in all hash tables, each block is chosen with probability
proportional to its length. The result follows by \eqref{ceb}.
\end{proof}

\begin{theorem}\label{TBmn}
For the exact model,
as $m,n\to\infty$ with $n/bm\to\ga$,
$B\mn\dto B_\ga$
with convergence of all moments.
\end{theorem}

\begin{proof}
By \refT{TBH}, $\hB\mn\dto\hB_\ga$ and thus by \eqref{tbhmn}
and \eqref{tbh},
\begin{equation}\label{tbx}
\frac{k\Pr(B\mn=k)}{\E B\mn}
=  \Pr(\hB\mn=k)
\to  \Pr(\hB_\ga=k)=\frac{k\Pr(B_\ga=k)}{\E B_\ga}.
\end{equation}
Furthermore, it is well-known that for integer-valued random variables,
convergence in distribution is equivalent to convergence in total
variation, \ie{} (in this case)
$\sum_k |\Pr(\hB\mn=k)-\Pr(\hB_\ga=k)|\to0$.
Consequently,
using \eqref{tbhmn} and \eqref{tbh} again,
\begin{equation}\label{tbxy}
\frac{1}{\E B\mn}
=\sumki\frac{1}k  \Pr(\hB\mn=k)
\to \sumki\frac{1}k \Pr(\hB_\ga=k)
=\frac{1}{\E B_\ga},
\end{equation}
and thus $\E B\mn\to\E B_\ga$. This and \eqref{tbx} yield 
$ \Pr(B\mn=k)\to  \Pr(B_\ga=k)$ for all $k\ge1$, \ie{} $B\mn\dto B_\ga$.
The moment convergence follows from the moment convergence in \refT{TBH},
since, generalizing \eqref{ebhb},
$\E B\mn^r = \E\hB\mn^{r-1}\E B\mn$ and
$\E B_\ga^r = \E\hB_\ga^{r-1}\E B_\ga$ 
for all $r$.
\end{proof}

\section{Unsuccessful search}\label{SU}
We consider the cost $U$ of an unsuccessful search. For convenience, we define
$U$ as the number of \emph{full} buckets that are searched, noting that the
total number of inspected buckets is $U+1$. 

Note also that $U$ is the
displacement of a new key inserted using the FCFS rule; 
thus $U$ can be seen as the cost
of inserting (and in the case of FCFS also retrieving) the
$n+1$st key. This approach is taken in Section \ref{SFC}.

\subsection{\Combinatorial} \label{SU-c}
Let 
\begin{equation}\label{eqUN}
  U(z,w,q) := \sum_{m\ge1}\sum_{n\ge0}
u_{m,n}(q) w^{bm} \frac{(mz)^n}{n!},
\end{equation}
where $u_{m,n}(q)$ is the probability generating function of the
cost $U$
of a unsuccessful search
in a hash table with $m$ buckets and
$n$ keys.
(We define $u_{m,n}(q)=0$ when $n\ge bm$.)

\begin{theorem}\label{TU1}
\begin{equation}\label{tu1}
  \begin{split}
U(bz,w,q) 
&= 
\Lambda_0(bz,w)
\frac{N_0(bz,w)-N_0(bz,wq^{1/b})}{1-q} 
\\&=
\frac{\prod_{j=0}^{b-1} \left(1-\frac{T(\omega^jzwq^{1/b})}{z}\right)
-\prod_{j=0}^{b-1} \left(1 - \frac{T(\omega^jzw)}{z}\right)}
{(1-q)\prod_{j=0}^{b-1} \left(1 - \frac{T(\omega^jzw)}{z}\right)}.	
  \end{split}
\end{equation}
\end{theorem}
\begin{proof}
$U(z,w,q)$ is  the trivariate generating function of the number of
hash tables (with $m$ buckets and $n$ keys)
where a new key added at a fixed bucket, say 0, ends up with
some displacement $k$.  (The variable $q$ marks this displacement.)
By symmetry, this number is the same as the number of hash tables
where a key added
to any bucket ends up in a fixed bucket, say the last, with
displacement $k$.
Since this implies that the last bucket is not full 
(before the addition of the new key), we can use the sequence construction
of hash tables in \refS{S:combin}; we then allow the new key to hash to any
bucket in the last cluster.
By \refR{RGen} and equations \eqref{Gen-mn1} and \eqref{Gen-mn2}, it is 
thus enough to study the problem in a cluster.

Let $\mathcal{C}$ be a combinatorial class representing a cluster.
In a  cluster with $m$ buckets, the number of visited full buckets  in a
unsuccessful search ranges from 0 to $m-1$. 
As a consequence, in the combinatorial model,
the specification $\Pos(\mathcal{C})$ 
represents the displacements in the last cluster;
by the argument above, the displacement  $U$ is thus represented by
$\Seq(\mathcal{C})*\Pos(\mathcal{C})$. 
By \eqref{Gen-mn1} and \refR{RGen},
this leads, using equations
\eqref{qq2} and \eqref{qq3},
to \eqref{tu1}.
\end{proof}

This result is also derived in \cite[Lemma 4.2]{Bigbuck}.

\subsection{\Probabilistic}
Let $U_i\ge0$ denote the number of
{full} buckets that we search
in an unsuccessful search for a key that
does not exist in the hash table, 
when we start with bucket $i$. 
Thus $U_i=k-i$ where $k$ is the index of the bucket that ends the block
containing $i$.

In the probabilistic version, 
we consider again the \infpoi{} on $\bbZ$, with $X_i\sim\Po(\ga b)$ independent.
Obviously, all $U_i$ have the same distribution, so we may take $i=0$; we
also use the notation $U_\ga=U_0$ for this model. We similarly use $U\mn$
for the exact model.
\begin{theorem}\label{TU}
In the \infpoi,
the probability generating function 
of\/ $U_\ga$ is
given by
\begin{equation}\label{psiu}
  \psiu(\zq) = \frac{\tba}{1-\zq}\prodlb\lrpar{1-\zetal(\zq)}.
\end{equation}

Moreover, for the exact model,
as $m,n\to\infty$ with $n/bm\to\ga$,
$U\mn\dto U_\ga$
with convergence of all moments;
furthermore, 
for some $\gd>0$, the \pgf{} converges to $\psiu(q)$,
uniformly for $|q|\le1+\gd$. 
\end{theorem}

\begin{proof}
This is similar to the proof of \refT{TBH}.
By the comments just made, $U_0=k$ if and only if there exists
$i\le-1$ such that $H_i<b$, and thus a block ends at $i$, and the block
beginning at $i+1$ ends at $k$, and thus has length $k-i$.
Consequently, using \refR{RB} and \eqref{ch<b},
\begin{equation}\label{puk}
  \begin{split}
\PP(U_0=k)&
=\sum_{i=-\infty}^{-1} \PP(H_i<b)\PP(B=k-i)	
=\cttx\sum_{i=-\infty}^{-1} \PP(B=k-i)	
\\&
= \cttx\PP(B>k),
\qquad k\ge0.
  \end{split}
\raisetag{\baselineskip}
\end{equation}
Hence, the \pgf{} $\psiu(\zq)$  is given by
\begin{equation}\label{pj}
  \begin{split}
\psiu(\zq)&:=\E \zq^{U_0}
=\sum_{k=0}^\infty \cttx\PP(B>k)\zq^k	
=\sum_{k=0}^\infty \sum_{j>k}\cttx\PP(B=j)\zq^k	
\\&\phantom:
=\cttx\sum_{j=1}^\infty\PP(B=j)\sum_{k=0}^{j-1}\zq^k	
=\cttx\sum_{j=1}^\infty\PP(B=j)\frac{1-\zq^j}{1-\zq}
\\&\phantom:
=\cttx\frac{1-\psib(\zq)}{1-\zq}.
  \end{split}
\raisetag{\baselineskip}
\end{equation}
The result \eqref{psiu} now follows by \eqref{psib}.  

As ususal, the final claim follows by \refT{Tconv} together with a uniform
estimate, which in this case comes from \refL{LexpB} and the 
bound $U_i\le \hB_i$.
\end{proof}

\begin{remark}
  Of course, $\psiu(1)=1$. We can verify that the \rhs{} of \eqref{psiu}
  equals 1 (as a limit) for $\zq=1$ by \eqref{zeta0'1} and \eqref{ctt}.

It is also possible to derive \eqref{psiu} (for $|q|<1$, say) from \refT{TU}
and \refT{Tconv}\ref{tconvres}.
\end{remark}

In principle, moments of $U_\ga$ can be computed by differentiation of
$\psi_U(q)$ in \eqref{psiu} at $q=1$. However, the factor $1-q$ in the
denominator makes the evaluation at $q=1$ complicated, since we have to take
a limit. Instead, we prefer to relate moments of $U_\ga$ to moments of $B_\ga$.
We give the expectation as an example, and leave higher moments to the reader.

\begin{corollary}\label{CU}
  As $m,n\to\infty$ with $n/bm\to\ga\in(0,1)$,
\begin{equation}\label{eu}
  \begin{split}
\E U\mn\to \E U_\ga 
&= 
\cttx \E\frac{B_\ga(B_\ga-1)}{2}
\\&
=
\frac{1}{2b(1-\ga)^2}
-\frac12
-  \frac{1}{b }
 \sum_{\ell=1}^{b-1}\frac{\zetal}{(1-\zetal)(1-\ga\zetal)}
.	
  \end{split}
\end{equation}
\end{corollary}

\begin{proof}
In the \infpoi,
  by \eqref{puk},
\begin{equation*}
\E U = \sumk k \Pr(U=k)	
=\cttx \sumk k \Pr(B>k)
= \cttx\sumj\Pr(B=j)\sum_{k=1}^{j-1} k
\end{equation*}
which yields the first equality in \eqref{eu}. The second follows by \refC{CB}.
\end{proof}

\begin{remark}\label{RUtot}
The cost of  an unsuccessful search starting at  $i$, 
measured as the total number of buckets inspected,
is  $U_i+1$, which has \pgf{} $q\psiu(q)$ and expectation $\E U+1$.
\end{remark}

The number of keys inspected in an unsuccessful search is 
$\tU_i=bU_i+U_i'$, where $U_i'$ is the number of keys in the first
non-full bucket. 
We can compute its \pgf{} too.
\begin{theorem}
  In the \infpoi, the number $\tU=\tU_\ga$ of keys inspected in an
  unsuccessful search has \pgf{}
\begin{equation}\label{psitu}
  \begin{split}
\psitu(\zq)
=\cttx\frac{\prodlb(\zq-\zetal(\zq^b))-\prodlb(\zq-\zetal(1))}{1-\zq^b}.
  \end{split}
\end{equation}

Moreover, for the exact model,
as $m,n\to\infty$ with $n/bm\to\ga$,
$\tU\mn\dto \tU_\ga$
with convergence of all moments;
furthermore, 
for some $\gd>0$, the \pgf{} converges to $\psitu(q)$,
uniformly for $|q|\le1+\gd$.
\end{theorem}

\begin{proof}
Arguing as in \eqref{pj}, we obtain the \pgf{}
\begin{equation*}
  \begin{split}
\psitu(\zq)&:=
\E \zq^{\tU_i}	
=\sum_{k=0}^\infty\sum_{\ell=0}^{b-1} \cttx\PP(B>k, Y_B=\ell)\zq^{bk+\ell}
\\&\phantom:
=\cttx\sum_{j=1}^\infty\sum_{\ell=0}^{b-1}\PP(B=j, Y_B=\ell)
\sum_{k=0}^{j-1}\zq^{bk+\ell}	
\\&\phantom:
=\cttx\sum_{j=1}^\infty\sum_{\ell=0}^{b-1}\PP(B=j, Y_B=\ell)
\frac{1-\zq^{bj}}{1-\zq^b}\zq^\ell
\\&\phantom:
=\cttx\frac{\E(\zq^{Y_B})-\E(\zq^{bB}\zq^{Y_B})}{1-\zq^b}
  \end{split}
\end{equation*}
and \eqref{psitu} follows by \eqref{psiby}.  

The final claims follow in the usual way from \refT{Tconv} and \refL{LexpB},
using $\tU_i< b\hB_i$.
\end{proof}

\begin{corollary}\label{CTU}
  As $m,n\to\infty$ with $n/bm\to\ga\in(0,1)$,
\begin{equation}\label{etu}
  \begin{split}
\E \tU\mn\to \E \tU_\ga 
&= 
\frac{1}{2(1-\ga)^2}
-\frac{b}{2}
+  
 \sum_{\ell=1}^{b-1}\frac{1-2\ga\zetal}{(1-\zetal)(1-\ga\zetal)}
.
  \end{split}
\end{equation}
\end{corollary}
\begin{proof}
  By differentiation of \eqref{psitu} at $q=1$, recalling $\zetal(1)=1$ and
  using \eqref{zetal'}--\eqref{zeta0'1} and \eqref{zetal''}.
We omit the details.
\end{proof}

\section{FCFS displacement}\label{SFC}

\subsection{\Combinatorial}
\label{CSFC}
We consider, as for Robin Hood hashing in \refS{SRHc}, the displacement of a
marked key $\bullet$, which we by symmetry may assume hashes to the first
bucket. 
Thus, let
\begin{equation}\label{eqFCFS}
  FCFS(z,w,q) := \sum_{m\ge1}\sum_{n\ge1}\sum_{k\ge0} 
FCFS_{m,n,k} w^{bm} \frac{z^n}{n!} q^k,
\end{equation}
where $FCFS_{m,n,k}$ is the number of hash tables of length $m$ with $n$
keys (one of them marked as $\bullet$) such that 
$\bullet$ hashes to the first bucket and
the displacement $\dfc$ of $\bullet$
equals $k$. For a given $m$ and $n$ with $1\le n\le bm$, there are $nm^{n-1}$
such  tables ($n$ choices to select $\bullet$ and $m^{n-1}$ choices to place 
the other $n-1$ elements). 
Thus, if 
$d_{m,n}(q)$ is the probability generating function for
the displacement of a random key in a hash table with $m$ buckets and
$n$ keys,
\begin{equation}\label{FCFS1}
  \begin{split}
  FCFS(z,w,q) 
&= \sum_{m\ge1}\sum_{n=1}^{bm} nm^{n-1}d_{m,n}(q) w^{bm} \frac{z^n}{n!}
\\&
= z\sum_{m\ge1}\sum_{n=0}^{bm-1} d_{m,n+1}(q) w^{bm} \frac{m^{n}z^n}{n!}.	
  \end{split}
\end{equation}
In other words, since there are $m^n$ hash tables with $m$ buckets and $n$ keys,
$z\qw FCFS(z,w,q)$ can be seen as the generating function of $d_{m,n+1}(q)$.

\begin{theorem}
\label{FCFSexact}
\begin{equation}\label{fcfs}
FCFS(bz,w,q)=b \int_0^z U(bt,we^{z-t},q) \,dt
\end{equation}
up to terms $z^n w^m q^k$ with $n> bm$.
\end{theorem}
\begin{proof}
The probability generating function for the displacement of a
random key when having $n$ keys in the table is
\begin{equation}\label{dmn}
d_{m,n}(q) = \frac1n \sum_{i=0}^{n-1} u_{m,i}(q),
\end{equation}
where the generating function of $u_{m,i}(q)$ is given by \eqref{eqUN}.
This assumes $n\le bm$, but for the rest of the proof we redefine
$d_{m,n}(q)$ so that \eqref{dmn} holds for all $m$ and $n\ge1$,
and we redefine $FCFS(z,w,q)$ by summing over all $n$ in \eqref{FCFS1}; 
this will only affect terms with $n>bm$.

By \eqref{dmn}, for all $m\ge1$ and $n\ge0$, 
\begin{equation*}
u_{m,n}(q) = (n+1)d_{m,n+1}(q)-nd_{m,n}(q),
\end{equation*}
and so, by \eqref{eqUN},
\begin{equation*}
U(bz,w,q) = \sum_{m\geq 1} w^{bm} \sum_{n \geq 0} \frac{(bmz)^n}{n!} 
\bigpar{ (n+1)d_{m,n+1}(q)-nd_{m,n}(q)}.
\end{equation*}
Thus,
$FCFS(bz,w,q)$ in \eqref{FCFS1}
satisfies
\begin{equation}
\ddz  FCFS(bz,w,q) - w\ddw FCFS(bz,w,q) = bU(bz,w,q).
\label{difeFCFS}
\end{equation}
The differential equation (\ref{difeFCFS}) together with the boundary 
condition 
$F(0,w,q)\allowbreak=0$
leads to the solution
\eqref{fcfs}.
\end{proof}

Note also that, by symmetry, in the definition of $FCFS(z,w,q)$, we may
instead of assuming that $\bullet$ hashes to the first bucket,  assume
that $\bullet$ ends up in a cluster that ends with a fixed bucket, for
example the last. 
We may then use the sequence construction in \refS{S:combin}, and by
\eqref{Gen-mn1} and \refR{RGen}, we obtain
\begin{equation}\label{fc-sequence}
  FCFS(z,w,q)=\gL_0(z,w)FC(z,w,q),
\end{equation}
where $FC(z,w,q)$ is the generating function for the displacements in an
almost full table.

\subsection{An alternative approach for $b=1$}

In this section we present an alternative approach to find the  
generating function for the displacement of a random key in 
FCFS, without considering unsuccessful searches.
We present it in detail only in the case $b=1$, when we are able to obtain
explicit generating functions, and give only a few remarks on the case $b>1$.

Thus, let $b=1$ and let $FC(z,w,q)$ be the generating function for the 
displacement of a marked key $\bullet$ in an almost full table 
(with $n=m-1$ keys).
Furthermore let $FC$ specify an almost full table with one key $\bullet$ marked
and $AF$ specify a normal almost full table. 

Consider an almost full table with $n>0$. Before inserting the last key, the
table consisted of two clusters (almost full subtables), with the last key
hashing to the first cluster. When considering also the marked key
$\bullet$, there are three cases, see \refF{decompFCFS}:
\begin{alphenumerate}
\item The new inserted key is $\bullet$ and has to find a position in an
  $AF$, with $q$ keeping track on the insertion cost.
\item The new inserted key is not $\bullet$ and has to find a position in
  \begin{enumerate}
  \item an $FC$ (in case $\bullet$ is in the first cluster),
  \item an $AF$ (in case $\bullet$ is in the second cluster).
  \end{enumerate}
\end{alphenumerate}

\begin{figure}[htb]
\begin{center}
\includegraphics[height=2.0cm]{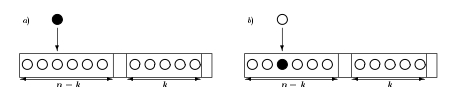}
\caption{Two cases in an FC table. 
In a) $\bullet$ is the last key inserted, 
while in b) $\bullet$ is already in the table when the last key is
inserted.}
\label{decompFCFS}
\medskip
\end{center}
\end{figure}

As a consequence we have the specification:
\begin{equation}
\label{specFCFS1}
FC = \Add\bigpar{\Pos_q(AF)*AF + \Pos(FC)*AF + \Pos(AF)*FC},
\end{equation}
where $\Pos_q$ means that we operate on the generating functions by $\HH$ in
\eqref{qq3}, while the normal $\Pos$ operates by $\frac{w}b\ddw$ as in 
\refFig{newcons}.
The construction $\Add$ translates to integration $\int$, see
\refS{q-calculus};
we eliminate this by taking the derivative $\ddzx$ on both sides. 
The specification \eqref{specFCFS1} thus yields the 
partial differential equation
{\multlinegap=0pt
\begin{multline}\label{eqdiffFCFS} 
\ddz FC(z,w,q) =
\HH[AF(z,w)]AF(z,w)
\\
+ w\ddw \bigpar{FC(z,w,q)}~AF(z,w) 
 +  w\ddw\bigpar{AF(z,w)}FC(z,w,q). 
\end{multline}}
Moreover, since all generating functions here are for almost full tables,
with $n=m-1$, they are all of the form $F(wz)/z$, and we have
$z\partial_z=w\partial_w-1$. 
Thus, \eqref{eqdiffFCFS} reduces to the ordinary differential equation
{\multlinegap=0pt
\begin{multline}\label{eqdiffFCFS2} 
\bigpar{1-zAF(z,w) }w\ddw FC(z,w,q) = 
  \Bigpar{1+ w\ddw\bigpar{zAF(z,w)}}FC(z,w,q)
\\
+z\HH[AF(z,w)]AF(z,w).
\end{multline}}
Furthermore, since $b=1$, \eqref{laFd} yields
$AF(z,w)=wF_0(zw)=N_0(z,w)$,
which by \eqref{Fd} or  \eqref{N0} yields 
$AF(w,z) =\xfrac{T(zw)}{z}$.
Substituting this in \eqref{eqdiffFCFS2}, 
the differential equation is easily solved 
using \eqref{T'} and the integrating factor $(1-T(zw))/T(zw)$,
and we obtain the following theorem.
\begin{theorem}\label{TFC1}
For $b=1$,
\begin{equation}\label{tfc1}
FC(z,w,q)=\frac{\bigpar{(1-T(zwq))^2-(1-T(zw))^2}T(zw)}{2z(1-q)(1-T(zw))}.
\end{equation}
\end{theorem}

Note that a similar derivation has been done in 
Section 4.6 of \cite{TubaT}. The difference is that in
\cite{TubaT} a recurrence for $n![z^nw^{n+1}]$ is derived, and
then moments are calculated. Even though it could have been
done, no attempt to find the generating function from the
recurrence is presented there. Moreover, in \cite{Exact}, the
same recurrence as in \cite{TubaT} is presented and then the
distribution is explicitly calculated for the Poisson model
(although the distribution is not presented for the
combinatorial model).

Here, with the same specification,
the symbolic method leads directly to the generating function,
avoiding the need to solve any recurrence. This example presents
in a very simple way how to use the approach 
\emph{``if you can specify it, you can analyze it''}.

By \eqref{fc-sequence} and \eqref{laLambda0},
we obtain immediately the generating function $FCFS$
(for general hash tables) in the case $b=1$.

\begin{theorem}
For $b=1$,
\begin{equation*}
FCFS(z,w,q)=\frac{((1-T(zwq))^2-(1-T(zw))^2)T(zw)}{2(1-q)(1-T(zw))(z-T(zw))}.
\end{equation*}
\end{theorem}

The distribution of the displacement $\dfc\mn$ can be obtained by extracting
coefficients in $FCFS$, see \eqref{eqFCFS}--\eqref{FCFS1}. 
Instead of doing this, we just quote a
result from Theorem 5.1
in  \cite{Exact} (see also \cite[Theorem 5.2]{SJ157}):
If $b=1$, $1\le n\le m$ and $k\ge0$, then
\begin{equation}
[q^k] d_{m,n}(q) = 1-\frac{n-1}{2m} - \sum_{j=0}^{k-1} 
\binom{n-1}{j} \frac{(j+1)^{j-2} (m-j-1)^{n-j-1}}{m^{n-1}}.
\end{equation}
Notice that there are two shifts from \cite{Exact}: the first one in
$k$ since in \cite{Exact} the studied random variable  is the
search cost and here it is the displacement; the second one in $n$,
since in \cite{Exact} the table has $n+1$ elements.

We can also derive a formula similar to \refT{TFC1} for completely full tables.
Let $FC_0(z,w,q)$ be the generating function for the 
displacement of a marked key $\bullet$ in a full table (with $n=m$ keys)
such that the last key is inserted in the last bucket.
(By rotational symmetry, we might instead require that any given key, or
$\bullet$, is inserted in any given bucket, or 
that any given key, or $\bullet$, hashes to a given bucket.)

\begin{theorem}\label{TFC0}
For $b=1$,
\begin{equation}\label{tfc0}
FC_0(z,w,q)=\frac{{(1-T(zwq))^2-(1-T(zw))^2}}{2(1-q)(1-T(zw))}.
\end{equation}
\end{theorem}

\begin{proof}
A full table is obtained by adding a key to an
almost full table. This leads to the specification
$FC_0=\Add(\Pos_q(AF)+\Pos(FC))$ and a differential equation that yields the
result.
However, we prefer instead the following combinatorial argument:
Say that there is a \emph{cut-point} after each bucket where the overflow is 0.
By the rotational symmetry, we may equivalently define $FC_0$ as the class
of full hash tables with a marked element $\bullet$ such that the first
cut-point after $\bullet$ comes at the end of the table.
By appending an almost full table to such a table, we obtain a table of the
type $FC$; this yields a specification $FC=FC_0*AF$, since this operation is
a bijcetion, with an inverse given by cutting a table of the type $FC$ at
the first cut-point after $\bullet$. Consequently,
\begin{equation}
FC(z,w,q)=FC_0(z,w,q)AF(z,w)
=FC_0(z,w,q)\frac{T(zw)}z
\end{equation}
and the result follows by \eqref{tfc1}.
\end{proof}

Moments are easily obtained from the generating functions above. In
particular, this gives another proof of \eqref{ammE}--\eqref{ammV}:

\begin{corollary}\label{CFC0}
  For $b=1$,
\begin{align}
  \E\dfcnn &= 
\frac{\sqrt{2\pi}}4n\qq -\frac{2}3+\frac{\sqrt{2\pi}}{48 n\qq} -\frac{2}{135n}
+O\bigpar{n^{-3/2}},
\\
  \Var(\dfcnn) &= \frac{\sqrt{2\pi}}{12}n^{3/2} 
+\lrpar{\frac{1}9-\frac{\pi}8}n 
+\frac{13\sqrt{2\pi}}{144}n^{1/2} 
-\frac{47}{405}-\frac{\pi}{48}
+O\bigpar{n^{-1/2}}.
\end{align}
\end{corollary}

\begin{proof}
  Taking $q=1$ in \eqref{tfc0} we obtain (by \lhopitals)
$T(zw)/(1-T(zw))=(1-T(zw))\qw-1$, which is the EGF of $n^n$;
this is correct since $FC_0$ counts
the full tables with a marked key such that the last key is inserted in the
last bucket, and there are indeed $n^n$ such tables.

Taking the first and second derivatives of \eqref{tfc0}
at $q=1$ we find, with $T=T(zw)$,
\begin{align*}
  \sumni n^n \E \dfcnn\frac{(zw)^n}{n!}
&=\Uq\partial_q FC_0(z,w,q)
=\frac{T^2}{2(1-T)^2}
\\&
=
\frac{1}{2(1-T)^2}
-
\frac{1}{1-T}
+\frac{1}2,
\\
 \sumni n^n \E (\dfcnn)^2\frac{(zw)^n}{n!}
&=\Uq\partial_q^2 FC_0(z,w,q)+\Uq\partial_q FC_0(z,w,q)
= \frac{3T^2-T^4}{6(1-T)^4}
\\&
=
\frac{1}{3(1-T)^4}
-\frac{1}{3(1-T)^3}
-\frac{1}{2(1-T)^2}
+\frac{2}{3(1-T)}
-\frac{1}{6}.
\end{align*}
Knuth and Pittel \cite{KP89} defined the tree polynomials $t_n(y)$ as the
coefficients in the expansion 
\begin{equation}
  \frac{1}{(1-T(z))^y} = \sumn t_n(y)\frac{z^n}{n!}.
\end{equation}
By identifying coefficients, we thus obtain the exact formulas
\begin{align}
 n^n \E\dfcnn &= \tfrac12 t_n(2)-t_n(1),
\\
n^n \E(\dfcnn)^2 
&= \tfrac13 t_n(4)-\tfrac13 t_n(3)-\tfrac12 t_n(2) +\tfrac23 t_n(1).
\end{align}
The result now follows from the asymptotic of $t_n(y)$ obtained by
singularity analysis, see \cite[(3.15) and the comments below it]{KP89}.
\end{proof}

\begin{remark}
The approach in this subsection can in principle be used also for $b>1$. 
For $1\le k\le b$, let $AF_k(z,w)$ be the
generating function for almost full hash tables with $k$ empty slots in the
last bucket, and let $FC_k(z,w,q)$ 
be the generating function for such tables with a marked key $\bullet$, with
$q$ marking the displacement of $\bullet$.
We can for each $k$ argue as for \eqref{specFCFS1}, but we now also have
the possibility that the last key hashes to an almost full table with $k+1$
empty slots (provided $k<b$).
We thus have the specifications
\begin{multline}
  FC_k = \Add\bigpar{\Pos_q(AF_1)*AF_k + \Pos(FC_1)*AF_k + \Pos(AF_1)*FC_k
\\
+\Pos_q(AF_{k+1})+\Pos(FC_{k+1})},
  \qquad k=1,\dots,b
\end{multline}
with $AF_{b+1}=FC_{b+1}=\emptyset$.
This yields a system of $b$ differential equations similar to
\eqref{eqdiffFCFS2} for the generating functions $FC_k(z,w,q)$; note  
that $AF_k$  is given by
$F_{b-k}$ in \eqref{laFd} and \eqref{Fd}.
However, we do not know any explicit solution to this system, and we therefore
do not consider it further, preferring the alternative approach presented
in Section \ref{CSFC}. (It seems possible that this approach at least might
lead to explicit generating functions for the expectation and higher
moments by differentiating the equations at $q=1$, but we have not pursued
this.)   
\end{remark}

\subsection{\Probabilistic}

In the probabilistic model, when inserting a new key in the hash table with
the FCFS rule,  
we do exactly as in an
unsuccessful search, except that at the end we insert the new key.
Hence the displacement of a new key has the same distribution as $U_i$
in \refS{SU}.
However (unlike the RH rule), the keys are never moved once they are
inserted, and when studying the displacement of a key already in the
table, we have to consider $U_i$ at the time the key was added.

We consider again infinite hashing on $\bbZ$, and add a time dimension by
letting the keys arrive to the buckets by independent Poisson process with
intensity 1. At time $t\ge0$, we thus have $X_i\sim\Po(t)$, so at time $\ga
b$ we have the same infinite Poisson model as before, but with each key
given an arrival 
time, with the arrival times being \iid{} and uniformly distributed on
$[0,\ga b]$. 
(We cannot proceed beyond time $t=b$; at this time the table becomes
full and an infinite number of keys overflow to $+\infty$; however, we
consider only $t<b$.)

Consider the table at time $\ga b$, containing all key with arrival
times in $[0,\ga b]$. 
We are interested in the FCFS displacement of a ``randomly chosen key''.
Since there is an infinite number of keys, this is not well-defined, but 
we can interpret it as follows (which gives the correct
limit of finite hash tables, see the theorem below): 
By a basic property of Poisson processes, 
if we condition on the existence of a key, $x$ say, that arrives to a
given bucket $i$ at a given time $t$, then all other keys form a Poisson
process with the same distribution as the original process. Hence the FCFS
displacement of $x$ has the same distribution as $U_\gb$, computed 
with the load factor $\ga$ replaced by $\gb:=t/b$.
Furthermore, as said above, the arrival times of the keys
are uniformly distributed in $[0,\ga b]$, so
$\gb$ is uniformly distributed in $[0,\ga]$.
Hence, the FCFS displacement $\dfc=\dfc_\ga$ of a random key is (formally by
definition) a random variable with the distribution
\begin{equation}\label{fcu}
  \PP(\dfc_\ga=k) = \frac1\ga\intoa\PP\bigpar{U_\gb=k}\dd\gb.
\end{equation}
This leads to the following, where we now
write $\ga$ as an explicit parameter of all quantities that depend on it.
\begin{theorem}
  \label{TFC}
In the \infpoi, 
the probability generating function $\psifc(\zq;\ga):=\E \zq^{\dfc}$ 
of\/
$\dfc=\dfc_\ga$ is 
given by   
\begin{equation}\label{psifc}
  \begin{split}
  \psifc(\zq;\ga)&
=\frac{1}{\ga}\intoa\psiu(\zq;\gb)\dd\gb
=\frac{1}{\ga}\intoa\frac{T_0(b\gb)}{1-\zq}\prodlb\bigpar{1-\zetal(\zq;\gb)}
\dd\gb
\\&
=\frac{1}{\ga}\intoa
\frac{b(1-\gb)\prodlb\bigpar{1-\zetal(\zq;\gb)}}
{(1-\zq)\prod_{\ell=1}^{b-1}(1-\zetal(1;\gb))}
\dd\gb.
  \end{split}
\end{equation}

Moreover, for the exact model,
as $m,n\to\infty$ with $n/bm\to\ga$,
$\dfc\mn\dto \dfc_\ga$
with convergence of all moments;
furthermore, 
for some $\gd>0$, the \pgf{} converges to $\psifc(q)$,
uniformly for $|q|\le1+\gd$.
\end{theorem}

\begin{proof}
The first equality in \eqref{psifc} follows by  \eqref{fcu}, and the second by
 \refT{TU}.

For the exact model, we have by the discussion above, for any $k\ge0$,
\begin{equation*}
  \Pr(\dfc\mn=k)=\frac{1}n\sum_{j=1}^n\Pr\bigpar{U_{m,j-1}=k}
=\int_0^1\Pr\bigpar{U_{m,\floor{nx}}=k}\dd x	.
\end{equation*}
For any $x>0$, $\floor{nx}/bm\dto x\ga$, and thus
$\Pr\bigpar{U_{m,\floor{nx}}=k}\to \Pr\bigpar{U_{x\ga}=k}$ by \refT{TU}.
Hence, by dominated convergence and the change of variables $\gb=x\ga$,
\begin{equation*}
  \Pr(\dfc\mn=k)
\to\int_0^1\Pr\bigpar{U_{x\ga}=k}\dd x
=\frac{1}{\ga}\int_0^{\ga}\Pr\bigpar{U_{\gb}=k}\dd \gb
=\Pr(\dfc_\ga=k),
\end{equation*}
by \eqref{fcu}. Hence $\dfc\mn\dto\dfc_\ga$  \asmn.
Convergence of moments and \pgf{} follows by \refL{LexpB} and
the estimate $\dfc\le U_i\le \hB_i$ for a key that hashes to $i$.
\end{proof}

\begin{corollary}\label{CFC}
As $m,n\to\infty$ with $n/bm\to\ga\in(0,1)$,
\begin{align}\label{efc}
\E \dfc_{m,n} &\to \E\dfc_\ga 
=\E\drh_\ga
= \frac{1}{2b\ga}\biggpar{\frac{1}{1-\ga}-b-b\ga}
+\frac{1}{b\ga}\sum_{\ell=1}^{b-1}\frac{1}{1-\zetal}.
\end{align}
\end{corollary}

\begin{proof}
  The convergence follows by \refT{TFC}. In the exact model,
  $\E\dfc\mn=\E\drh\mn$, since the total displacement does not depend on
  the insertion rule, see \refL{LED}. Hence, using also \refC{CRH},
$\E\dfc_\ga=\E\drh_\ga$, and the result follows by \eqref{ed}.
Alternatively, by \eqref{fcu},
\begin{equation}\label{efcu}
  \E\dfc_\ga = \frac1\ga\intoa\E U_\gb\dd\gb,
\end{equation}
where $\E U_\gb$ is given by \eqref{eu}. It can be verified that this yields
\eqref{efc} (simplest by showing that
$\frac{\ddx}{\ddx\ga}\bigpar{\ga\E\drh_\ga}=\E U_\ga$, 
using \eqref{ed}, \eqref{zetal}
and \eqref{T'}).
\end{proof}

\section{Some experimental results} \label{SER}

In this section we check our theoretical results against the
experimental values 
for the FCFS heuristic
presented in the original paper
\cite{Peterson} from 1957, 
where the linear probing hashing algorithm was first presented,
since this is the first distributional analysis made for the 
problem (for the general case, $b\geq 1$).

The historical value of this section, relies on the fact that
this is the first time that these original experimental values
are checked against distributional theoretical results.

\begin{center}
\begin{figure}[htb!]
\begin{center}
\begin{tabular}{|r|r|r|c|c|}
\hline
 Length & \# records & Length $\times$ \#  & Empirical prob. & Theoretical 
\\
\hline
  1 & 8418 &  8418  & 0.9353 &  0.9352\\
  2 &  336 &   672  & 0.0373 &  0.0364 \\
  3 &  111 &   333  & 0.0123 &  0.0122\\
  4 &   70 &   280  & 0.0078 &  0.0059\\
  5 &   26 &   130  & 0.0029 &  0.0033\\
  6 &    9 &    54  & 0.0010 &  0.0020\\
  7 &   14 &    98  & 0.0016 &  0.0013\\
  8 &    7 &    56  & 0.0008 &  0.0009\\
  9 &    5 &    45  & 0.0006 &  0.0006\\
 10 &    1 &    10  & 0.0001 &  0.0005\\
 11 &    1 &    11  & 0.0001 &  0.0003\\
 12 &    0 &     0  & 0.0000 &  0.0003\\
 13 &    1 &    13  & 0.0001 &  0.0002\\
 14 &    1 &    14  & 0.0001 &  0.0002\\
\hline
sum$|$average   & 9000 & 10134  & 1.1260 & 1.1443\\
\hline
\end{tabular}
\caption{Table 3 in \cite{Peterson}, with $b=20$, $m=10000$ and
$n=9000$ ($\alpha=0.9$),
together with theoretical results taken from
our distributional results in the Poisson Model for the FCFS
heuristic, \refT{TFC}. For the keys with low search cost (3 or less)
there is a very good agreement with the experimental results.
This is consistent with the explanation that the high variance
is originated by all the latest keys inserted (and as a
consequence, in FCFS, by the ones with larger displacement). The
last line presents the average search cost, where a difference
with the theoretical result in \refC{CFC} is noticed.}
\end{center}
\end{figure}
\end{center}

Even though we have exact distributional results for the search
cost of random keys 
by \refT{FCFSexact} for the generating function,
the coefficients are very difficult to
extract and to calculate. As a consequence, except for the case
when $b=1$ where exact results are easy to use by means of
\eqref{average}, we use the asymptotic
approximate results in the Poisson Model derived in 
\refT{TFC} and \refC{CFC}. Notice that \refT{TFC} and \refC{CFC}
evaluate the \emph {displacement} of a random key, while in
\cite{Peterson} its \emph {search cost} is calculated. As a
consequence, we have to add 1 to the theoretical results.
When the table is full, we use \eqref{asymptRH}. 
\begin{center}
\begin{figure}[htbp!]
\begin{center}
\begin{tabular}{|c|c|c|c|c|c|}
\hline
\% Full & 1st run & 2nd run &
3rd run & 4th run & Theoretical
\\
\hline
 40 & 1.000 & 1.000 &  1.000 & 1.000 & 1.000 \\
 60 & 1.001 & 1.002 &  1.002 & 1.003 & 1.002 \\
 70 & 1.008 & 1.013 &  1.009 & 1.010 & 1.010 \\
 80 & 1.026 & 1.043 &  1.029 & 1.035 & 1.036 \\
 85 & 1.064 & 1.073 &  1.062 & 1.067 & 1.069 \\
 90 & 1.134 & 1.126 &  1.138 & 1.137 & 1.144 \\
 95 & 1.321 & 1.284 &  1.331 & 1.392 & 1.386 \\
 97 & 1.623 & 1.477 &  1.512 & 1.797 & 1.716 \\
 99 & 2.944 & 2.112 &  2.085 & 2.857 & 3.379 \\
100 & 4.735 & 3.319 &  3.830 & 4.299 & 4.002 \\
\hline
\end{tabular}
\caption{Table 4 in \cite{Peterson}: average length of search of
a random record with $b=20$ and $m=500$, together with theoretical results. 
(The experiments were carried four times over the same data.) 
As $\alpha$ increases, 
so does the variance of the results for the four runs.}
\end{center}
\end{figure}
\end{center}

All the simulations in \cite{Peterson} 
were done in random-access memory with a IBM
704 addressing asystem (simulated memory with random
identification numbers). It is interesting to note that the IBM
704 was the first mass-produced computer with floating-point
arithmetic hardware. It was introduced by IBM in 1954.
\begin{center}
\begin{figure}[htbp!]
\begin{center}
\resizebox{\textwidth}{!}{
\begin{tabular}{|c|c|c|c|c|c|c|c|}
\hline
 & $b=1~ |~ m=500$ &  & & &
 & $b=2~ |~ m=500$ & 
\\
\hline
\% Full & Experimental & Poisson & Exact & &
\% Full & Experimental & Poisson
\\
\hline
 10 & 1.053 & 1.056 & 1.055 & & 20 & 1.034 & 1.024 \\
 20 & 1.137 & 1.125 & 1.125 & & 40 & 1.113 & 1.103 \\
 30 & 1.230 & 1.214 & 1.213 & & 60 & 1.325 & 1.293 \\
 40 & 1.366 & 1.333 & 1.331 & & 70 & 1.517 & 1.494 \\
 50 & 1.541 & 1.500 & 1.496 & & 80 & 1.927 & 1.903 \\
 60 & 1.823 & 1.750 & 1.741 & & 90 & 3.148 & 3.147 \\
 70 & 2.260 & 2.167 & 2.142 & & 95 & 5.112 & 5.644 \\
 80 & 3.223 & 3.000 & 2.911 & &100 &11.389 &10.466 \\
 90 & 5.526 & 5.500 & 4.889 & &    &       &       \\
100 &16.914 &  ---  &14.848 & &    &       &       \\
\hline
\end{tabular}
} 
\caption{Table 5B in \cite{Peterson}:
average length of search with $b=1$ or 2 and $m=500$
(average of 9 runs over the same data), 
together with exact results taken
from equation \eqref{average} and the Poisson approximation
taken from \refC{CFC} (and \eqref{asymptRH} 
for full tables). 
Notice that the Poisson approximation
matches very well when the load factor is less than $0.9$. This
is consistent with the fact that this approximation largely
overestimates the exact result when $\alpha \to 1$. 
 }
\end{center}
\end{figure}
\end{center}

The high variance of the FCFS heuristic 
is originated by the fact that the first
keys stay in their home location, but when collisions start to
appear, the search cost of the latest inserted keys
increases very rapidly. For example, the last keys inserted
have an $O(m)$ displacement in average. As an aside, for the
case $b=1$ there are very interesting results that analyze the
way contiguous clusters coalesce \cite{ChaLou02}. A good
understanding of this process, leads to a rigurous explanation
of this problem.
\begin{center}
\begin{figure}[htbp!]
\begin{center}
\resizebox{\textwidth}{!}{
\begin{tabular}{|c|c|c|c|c|c|c|}
\hline
\# Runs & 8 & 7 & 4 & 4 & 3 & 3 \\
\hline
\% Full & b=5 & b=10 & b=20 & b=30 & b=40 & b = 50  \\
        & bm=2500 & bm=5000 & bm=5000 & bm=10000 & bm=10000 & bm=10000  \\
\hline
 40 & 1.015 (1.012) & 1.001 (1.001) &  1.000 (1.000) & 1.000 (1.000) & 1.000 (1.000) & 1.000 (1.000) \\
 60 & 1.072 (1.066) & 1.016 (1.015) &  1.002 (1.002) & 1.001 (1.000) & 1.000 (1.000) & 1.000 (1.000) \\
 70 & 1.131 (1.136) & 1.042 (1.042) &  1.010 (1.010) & 1.003 (1.003) & 1.001 (1.001) & 1.000 (1.000) \\
 80 & 1.280 (1.289) & 1.111 (1.110) &  1.033 (1.036) & 1.017 (1.017) & 1.011 (1.009) & 1.005 (1.005) \\
 85 & 1.443 (1.450) & 1.172 (1.186) &  1.066 (1.069) & 1.038 (1.036) & 1.028 (1.022) & 1.015 (1.015) \\
 90 & 1.762 (1.777) & 1.330 (1.345) &  1.134 (1.144) & 1.082 (1.083) & 1.071 (1.055) & 1.034 (1.040) \\
 95 & 2.467 (2.771) & 1.755 (1.837) &  1.334 (1.386) & 1.231 (1.241) & 1.185 (1.171) & 1.110 (1.130) \\
 97 & 3.154 (4.102) & 2.187 (2.501) &  1.602 (1.716) & 1.374 (1.460) & 1.399 (1.334) & 1.228 (1.260) \\
 99 & 4.950 (10.766) & 3.212 (5.831) &  2.499 (3.379) & 1.852 (2.567) & 2.007 (2.164) & 1.585 (1.923) \\
100 & 6.870 (6.999) & 4.889 (5.244) &  4.041 (4.002) & 2.718 (3.451) & 2.844 (3.123) & 2.102 (2.899) \\
\hline
\end{tabular}
} 
\caption{Table 5A in \cite{Peterson}: average length of search. 
The first number is the
experimental result, and in parenthesis the theoretical value.
We see again that the disagreement is larger when $\alpha$
is close to 1.}
\end{center}
\end{figure}
\end{center}

As a consequence, as it is explicitly said in \cite{Peterson},
the experiments had to be performed several times, and an
average of the results are presented. Our theoretical results
seem to agree well with the experimental values when the load factor 
is less than $\alpha = 0.9$.
Moreover, the closer the value of $\alpha$ gets to
1 then the larger the difference. This is expected, since the
formulae are good approximations when $\alpha < 1$, and tend to
$\infty$ when $\alpha \to 1$.

\begin{acks}
Philippe Flajolet has had a strong influence in our scientific careers.
The core of the use of the symbolic method in hashing problems has been
taken from \cite{Flajolet:slides}. Thank you Philippe 
for all the work you have left to inspire our research.
We also thank Alois Panholzer for interesting discussions, and
Hsien-Kuei Hwang for suggesting us the derivation that leads to
Theorem \ref{FCFSexact}.
\end{acks}

\newcommand\RSA{\emph{Random Struct. Alg.} }
\newcommand\DMTCS{\jour{Discr. Math. Theor. Comput. Sci.} }

\newcommand\AMS{Amer. Math. Soc.}
\newcommand\Springer{Springer-Verlag}
\newcommand\Wiley{Wiley}

\newcommand\vol{\textbf}
\newcommand\jour{\emph}
\newcommand\book{\emph}
\newcommand\inbook{\emph}
\def\no#1#2,{\unskip#2, no. #1,} 
\newcommand\toappear{\unskip, to appear}

\newcommand\urlsvante{\url{http://www.math.uu.se/~svante/papers/}}
\newcommand\arxiv[1]{\url{arXiv:#1.}}
\newcommand\arXiv{\arxiv}

\end{document}

%% file: rh.tex
\setlength{\unitlength}{0.00043300in}%
\begingroup\makeatletter\ifx\SetFigFont\undefined%
\gdef\SetFigFont#1#2#3#4#5{%
  \reset@font\fontsize{#1}{#2pt}%
  \fontfamily{#3}\fontseries{#4}\fontshape{#5}%
  \selectfont}%
\fi\endgroup%
\begin{picture}(6024,824)(-189,-673)
\thicklines
\put(3001,-61){\line( 0,-1){1200}}
\put(3601,-61){\line( 0,-1){1200}}
\put(4201,-61){\line( 0,-1){1200}}
\put(4801,-61){\line( 0,-1){1200}}
\put(5401,-61){\line( 0,-1){1200}}
\put(6001,-61){\line( 0,-1){1200}}
\put(6601,-61){\line( 0,-1){1200}}
\put(7201,-61){\line( 0,-1){1200}}
\put(7801,-61){\line( 0,-1){1200}}
\put(2401,-661){\framebox(6000,600){}}
\put(2401,-1261){\framebox(6000,600){}}
\put(1901,-736){\makebox(0,0)[lb]{\smash{\SetFigFont{12}{14.4}{\rmdefault}{\mddefault}{\updefault}${\textstyle a}$}}}
\put(2501,-436){\makebox(0,0)[lb]{\smash{\SetFigFont{12}{14.4}{\rmdefault}{\mddefault}{\updefault}${\textstyle 69}$}}}
\put(3101,-436){\makebox(0,0)[lb]{\smash{\SetFigFont{12}{14.4}{\rmdefault}{\mddefault}{\updefault}${\textstyle 10}$}}}
\put(3751,-436){\makebox(0,0)[lb]{\smash{\SetFigFont{12}{14.4}{\rmdefault}{\mddefault}{\updefault}${\textstyle }$}}}
\put(4351,-436){\makebox(0,0)[lb]{\smash{\SetFigFont{12}{14.4}{\rmdefault}{\mddefault}{\updefault}${\textstyle   }$}}}
\put(4951,-436){\makebox(0,0)[lb]{\smash{\SetFigFont{12}{14.4}{\rmdefault}{\mddefault}{\updefault}${\textstyle 24}$}}}
\put(5551,-436){\makebox(0,0)[lb]{\smash{\SetFigFont{12}{14.4}{\rmdefault}{\mddefault}{\updefault}${\textstyle   }$}}}
\put(6151,-436){\makebox(0,0)[lb]{\smash{\SetFigFont{12}{14.4}{\rmdefault}{\mddefault}{\updefault}${\textstyle 36}$}}}
\put(6751,-436){\makebox(0,0)[lb]{\smash{\SetFigFont{12}{14.4}{\rmdefault}{\mddefault}{\updefault}${\textstyle 77}$}}}
\put(7351,-436){\makebox(0,0)[lb]{\smash{\SetFigFont{12}{14.4}{\rmdefault}{\mddefault}{\updefault}${\textstyle 18}$}}}
\put(7951,-436){\makebox(0,0)[lb]{\smash{\SetFigFont{12}{14.4}{\rmdefault}{\mddefault}{\updefault}${\textstyle 78}$}}}
\put(2501,-1066){\makebox(0,0)[lb]{\smash{\SetFigFont{12}{14.4}{\rmdefault}{\mddefault}{\updefault}${\textstyle 79}$}}}
\put(3101,-1066){\makebox(0,0)[lb]{\smash{\SetFigFont{12}{14.4}{\rmdefault}{\mddefault}{\updefault}${\textstyle }$}}}
\put(3701,-1066){\makebox(0,0)[lb]{\smash{\SetFigFont{12}{14.4}{\rmdefault}{\mddefault}{\updefault}${\textstyle   }$}}}
\put(4301,-1066){\makebox(0,0)[lb]{\smash{\SetFigFont{12}{14.4}{\rmdefault}{\mddefault}{\updefault}${\textstyle   }$}}}
\put(4901,-1066){\makebox(0,0)[lb]{\smash{\SetFigFont{12}{14.4}{\rmdefault}{\mddefault}{\updefault}${\textstyle   }$}}}
\put(5501,-1066){\makebox(0,0)[lb]{\smash{\SetFigFont{12}{14.4}{\rmdefault}{\mddefault}{\updefault}${\textstyle   }$}}}
\put(6101,-1066){\makebox(0,0)[lb]{\smash{\SetFigFont{12}{14.4}{\rmdefault}{\mddefault}{\updefault}${\textstyle 56 }$}}}
\put(6701,-1066){\makebox(0,0)[lb]{\smash{\SetFigFont{12}{14.4}{\rmdefault}{\mddefault}{\updefault}${\textstyle 97}$}}}
\put(7301,-1066){\makebox(0,0)[lb]{\smash{\SetFigFont{12}{14.4}{\rmdefault}{\mddefault}{\updefault}${\textstyle 38}$}}}
\put(7901,-1066){\makebox(0,0)[lb]{\smash{\SetFigFont{12}{14.4}{\rmdefault}{\mddefault}{\updefault}${\textstyle 49}$}}}
\put(2601,-1666){\makebox(0,0)[lb]{\smash{\SetFigFont{12}{14.4}{\rmdefault}{\mddefault}{\updefault}${\textstyle 0}$}}}
\put(3201,-1666){\makebox(0,0)[lb]{\smash{\SetFigFont{12}{14.4}{\rmdefault}{\mddefault}{\updefault}${\textstyle 1}$}}}
\put(3801,-1666){\makebox(0,0)[lb]{\smash{\SetFigFont{12}{14.4}{\rmdefault}{\mddefault}{\updefault}${\textstyle 2}$}}}
\put(4401,-1666){\makebox(0,0)[lb]{\smash{\SetFigFont{12}{14.4}{\rmdefault}{\mddefault}{\updefault}${\textstyle 3}$}}}
\put(5001,-1666){\makebox(0,0)[lb]{\smash{\SetFigFont{12}{14.4}{\rmdefault}{\mddefault}{\updefault}${\textstyle 4}$}}}
\put(5601,-1666){\makebox(0,0)[lb]{\smash{\SetFigFont{12}{14.4}{\rmdefault}{\mddefault}{\updefault}${\textstyle 5}$}}}
\put(6201,-1666){\makebox(0,0)[lb]{\smash{\SetFigFont{12}{14.4}{\rmdefault}{\mddefault}{\updefault}${\textstyle 6}$}}}
\put(6801,-1666){\makebox(0,0)[lb]{\smash{\SetFigFont{12}{14.4}{\rmdefault}{\mddefault}{\updefault}${\textstyle 7}$}}}
\put(7401,-1666){\makebox(0,0)[lb]{\smash{\SetFigFont{12}{14.4}{\rmdefault}{\mddefault}{\updefault}${\textstyle 8}$}}}
\put(8001,-1666){\makebox(0,0)[lb]{\smash{\SetFigFont{12}{14.4}{\rmdefault}{\mddefault}{\updefault}${\textstyle 9}$}}}
\end{picture}
\vspace{1cm}

%% file: rh1.tex
\setlength{\unitlength}{0.00043300in}%
\begingroup\makeatletter\ifx\SetFigFont\undefined%
\gdef\SetFigFont#1#2#3#4#5{%
  \reset@font\fontsize{#1}{#2pt}%
  \fontfamily{#3}\fontseries{#4}\fontshape{#5}%
  \selectfont}%
\fi\endgroup%
\begin{picture}(6024,824)(-189,-673)
\thicklines
\put(3001,-61){\line( 0,-1){1200}}
\put(3601,-61){\line( 0,-1){1200}}
\put(4201,-61){\line( 0,-1){1200}}
\put(4801,-61){\line( 0,-1){1200}}
\put(5401,-61){\line( 0,-1){1200}}
\put(6001,-61){\line( 0,-1){1200}}
\put(6601,-61){\line( 0,-1){1200}}
\put(7201,-61){\line( 0,-1){1200}}
\put(7801,-61){\line( 0,-1){1200}}
\put(2401,-661){\framebox(6000,600){}}
\put(2401,-1261){\framebox(6000,600){}}
\put(1901,-736){\makebox(0,0)[lb]{\smash{\SetFigFont{12}{14.4}{\rmdefault}{\mddefault}{\updefault}${\textstyle a}$}}}
\put(2501,-436){\makebox(0,0)[lb]{\smash{\SetFigFont{12}{14.4}{\rmdefault}{\mddefault}{\updefault}${\textstyle 49}$}}}
\put(3101,-436){\makebox(0,0)[lb]{\smash{\SetFigFont{12}{14.4}{\rmdefault}{\mddefault}{\updefault}${\textstyle 79}$}}}
\put(3751,-436){\makebox(0,0)[lb]{\smash{\SetFigFont{12}{14.4}{\rmdefault}{\mddefault}{\updefault}${\textstyle }$}}}
\put(4351,-436){\makebox(0,0)[lb]{\smash{\SetFigFont{12}{14.4}{\rmdefault}{\mddefault}{\updefault}${\textstyle   }$}}}
\put(4951,-436){\makebox(0,0)[lb]{\smash{\SetFigFont{12}{14.4}{\rmdefault}{\mddefault}{\updefault}${\textstyle 24}$}}}
\put(5551,-436){\makebox(0,0)[lb]{\smash{\SetFigFont{12}{14.4}{\rmdefault}{\mddefault}{\updefault}${\textstyle   }$}}}
\put(6151,-436){\makebox(0,0)[lb]{\smash{\SetFigFont{12}{14.4}{\rmdefault}{\mddefault}{\updefault}${\textstyle 36}$}}}
\put(6751,-436){\makebox(0,0)[lb]{\smash{\SetFigFont{12}{14.4}{\rmdefault}{\mddefault}{\updefault}${\textstyle 77}$}}}
\put(7351,-436){\makebox(0,0)[lb]{\smash{\SetFigFont{12}{14.4}{\rmdefault}{\mddefault}{\updefault}${\textstyle 18}$}}}
\put(7951,-436){\makebox(0,0)[lb]{\smash{\SetFigFont{12}{14.4}{\rmdefault}{\mddefault}{\updefault}${\textstyle 58}$}}}
\put(2501,-1066){\makebox(0,0)[lb]{\smash{\SetFigFont{12}{14.4}{\rmdefault}{\mddefault}{\updefault}${\textstyle 69}$}}}
\put(3101,-1066){\makebox(0,0)[lb]{\smash{\SetFigFont{12}{14.4}{\rmdefault}{\mddefault}{\updefault}${\textstyle 10}$}}}
\put(3701,-1066){\makebox(0,0)[lb]{\smash{\SetFigFont{12}{14.4}{\rmdefault}{\mddefault}{\updefault}${\textstyle   }$}}}
\put(4301,-1066){\makebox(0,0)[lb]{\smash{\SetFigFont{12}{14.4}{\rmdefault}{\mddefault}{\updefault}${\textstyle   }$}}}
\put(4901,-1066){\makebox(0,0)[lb]{\smash{\SetFigFont{12}{14.4}{\rmdefault}{\mddefault}{\updefault}${\textstyle   }$}}}
\put(5501,-1066){\makebox(0,0)[lb]{\smash{\SetFigFont{12}{14.4}{\rmdefault}{\mddefault}{\updefault}${\textstyle   }$}}}
\put(6101,-1066){\makebox(0,0)[lb]{\smash{\SetFigFont{12}{14.4}{\rmdefault}{\mddefault}{\updefault}${\textstyle 56 }$}}}
\put(6701,-1066){\makebox(0,0)[lb]{\smash{\SetFigFont{12}{14.4}{\rmdefault}{\mddefault}{\updefault}${\textstyle 97}$}}}
\put(7301,-1066){\makebox(0,0)[lb]{\smash{\SetFigFont{12}{14.4}{\rmdefault}{\mddefault}{\updefault}${\textstyle 38}$}}}
\put(7901,-1066){\makebox(0,0)[lb]{\smash{\SetFigFont{12}{14.4}{\rmdefault}{\mddefault}{\updefault}${\textstyle 78}$}}}
\put(2601,-1666){\makebox(0,0)[lb]{\smash{\SetFigFont{12}{14.4}{\rmdefault}{\mddefault}{\updefault}${\textstyle 0}$}}}
\put(3201,-1666){\makebox(0,0)[lb]{\smash{\SetFigFont{12}{14.4}{\rmdefault}{\mddefault}{\updefault}${\textstyle 1}$}}}
\put(3801,-1666){\makebox(0,0)[lb]{\smash{\SetFigFont{12}{14.4}{\rmdefault}{\mddefault}{\updefault}${\textstyle 2}$}}}
\put(4401,-1666){\makebox(0,0)[lb]{\smash{\SetFigFont{12}{14.4}{\rmdefault}{\mddefault}{\updefault}${\textstyle 3}$}}}
\put(5001,-1666){\makebox(0,0)[lb]{\smash{\SetFigFont{12}{14.4}{\rmdefault}{\mddefault}{\updefault}${\textstyle 4}$}}}
\put(5601,-1666){\makebox(0,0)[lb]{\smash{\SetFigFont{12}{14.4}{\rmdefault}{\mddefault}{\updefault}${\textstyle 5}$}}}
\put(6201,-1666){\makebox(0,0)[lb]{\smash{\SetFigFont{12}{14.4}{\rmdefault}{\mddefault}{\updefault}${\textstyle 6}$}}}
\put(6801,-1666){\makebox(0,0)[lb]{\smash{\SetFigFont{12}{14.4}{\rmdefault}{\mddefault}{\updefault}${\textstyle 7}$}}}
\put(7401,-1666){\makebox(0,0)[lb]{\smash{\SetFigFont{12}{14.4}{\rmdefault}{\mddefault}{\updefault}${\textstyle 8}$}}}
\put(8001,-1666){\makebox(0,0)[lb]{\smash{\SetFigFont{12}{14.4}{\rmdefault}{\mddefault}{\updefault}${\textstyle 9}$}}}
\end{picture}
\vspace{1cm}